\documentclass[reqno]{amsart}

\usepackage{amsfonts, amscd, epsfig, amsmath, amssymb,enumerate}
\usepackage{amstext}
\usepackage{graphicx}
\usepackage[usenames,dvipsnames,svgnames,table]{xcolor}
\usepackage{mathrsfs}
\usepackage{mathptmx}       
\usepackage{helvet}         
\usepackage{courier}        
\usepackage{physics}
\definecolor{brightcerulean}{rgb}{0.11, 0.67, 0.84}
\definecolor{cerulean}{rgb}{0.0, 0.48, 0.65}
\definecolor{Gray}{rgb}{0.5, 0.5, 0.5}
\definecolor{brightcerulean}{rgb}{0.11, 0.67, 0.84}
\definecolor{cerulean}{rgb}{0.0, 0.48, 0.65}
\definecolor{Gray}{rgb}{0.5, 0.5, 0.5}

\usepackage{a4wide}
\definecolor{columbiablue}{rgb}{0.61, 0.87, 1.0}

\definecolor{azulzinho}{rgb}{0.61, 0.87, 1.0}

\usepackage{tikz}
\usetikzlibrary{arrows,decorations.pathmorphing,backgrounds,positioning,fit,petri}
\usetikzlibrary{shapes}
\usetikzlibrary{arrows.meta}

\newtheorem{thm}{Theorem}[section]
\newtheorem{lem}[thm]{Lemma}
\newtheorem{cor}[thm]{Corollary}

\newtheorem{prop}[thm]{Proposition}

\newtheorem{definition}[thm]{Definition}

\newtheorem{rem}[thm]{Remark}

\newcommand\eps{\epsilon}

\newcommand\EE{{\mathbb E}}
\newcommand\PP{{\mathbb P}}

\newcommand{\inner}[1]{\langle #1 \rangle}

\newcommand{\mc}[1]{{\mathcal #1}}

\newcommand{\bb}[1]{{\mathbb #1}}
\newcommand{\ccl}[1]{{\color{red} #1}}

\usepackage[charter]{mathdesign}
\usepackage{cancel}

\usepackage{tikz}
\usetikzlibrary{backgrounds}
\usetikzlibrary{patterns,fadings}
\usetikzlibrary{arrows,decorations.pathmorphing}
\usetikzlibrary{calc}
\definecolor{light-gray}{gray}{0.95}
\usepackage{float}
\usepackage[colorlinks=true,linkcolor=blue,citecolor=magenta]{hyperref}

\renewcommand{\leq}{\leqslant}

\renewcommand{\geq}{\geqslant}
\renewcommand{\ge}{\geqslant}

\newcommand{\pa}[1]{\left(#1\right)}

\let\oldtocsection=\tocsection
\let\oldtocsubsection=\tocsubsection
\let\oldtocsubsubsection=\tocsubsubsection
\renewcommand{\tocsection}[2]{\hspace{0em}\oldtocsection{#1}{#2}}
\renewcommand{\tocsubsection}[2]{\hspace{1em}\oldtocsubsection{#1}{#2}}
\renewcommand{\tocsubsubsection}[2]{\hspace{2em}\oldtocsubsubsection{#1}{#2}}
\DeclareRobustCommand{\SkipTocEntry}[5]{}
\title[]{Hydrodynamics for  SSEP with non-reversible\\ slow  boundary dynamics: Part I,  the critical regime and beyond}
\author{C. Erignoux}
\address{Cl\'ement Erignoux, Equipe PARADYSE, Bureau B211
Centre INRIA Lille Nord-Europe
Park Plaza, Parc scientifique de la Haute-Borne, 40 Avenue Halley B\^atiment B, 59650 Villeneuve-d'Ascq
France}
\email{{\tt clement.erignoux@inria.fr}}

\author{P.  Gon\c calves}
\address{Patr\'icia Gon\c calves, Center for Mathematical Analysis,  Geometry and Dynamical Systems,
Instituto Superior T\'ecnico, Universidade de Lisboa,
Av. Rovisco Pais, 1049-001 Lisboa, Portugal.}
\email{{\tt pgoncalves@tecnico.ulisboa.pt}}

\author{G. Nahum}
\address{Gabriel Nahum, Center for Mathematical Analysis,  Geometry and Dynamical Systems, Instituto Superior T\'ecnico, Universidade de Lisboa,
	Av. Rovisco Pais, 1049-001 Lisboa, Portugal.}
\email{{\tt gabriel.nahum@tecnico.ulisboa.pt}}

\thanks{C.E. gratefully acknowledges funding from the European Research Council under the European Unions Horizon 2020 Program, ERC Consolidator GrantUniCoSM (grant agreement no 724939). P.G. thanks FCT/Portugal for support through the project UID/MAT/04459/2013. { G.N thanks FCT/Portugal for the support through the project Lisbon Mathematics PhD (LisMath).} This project has received funding from the European Research Council (ERC) under  the European Union's Horizon 2020 research and innovative programme 
(grant agreement   No 715734).}

\date{\today.}

\begin{document}

\begin{abstract}
The purpose of this article is to provide a simple proof of the hydrodynamic and hydrostatic behavior of the SSEP in contact with slowed reservoirs which inject and remove particles in a finite size windows at the extremities of the bulk. More precisely, the reservoirs inject/remove particles at/from any point of a window of size $K$ placed at each extremity of the bulk and particles are injected/removed to the first open/occupied position in that window.  The  hydrodynamic limit is  given by the heat equation with non-linear Robin boundary conditions or Neumann boundary conditions, the latter being in the case when the reservoirs are too slow.  The proof goes through the entropy method of \cite{GPV}.  We also derive the  hydrostatic limit  for this model, whose proof  is based on the method developed in \cite{LT} and \cite{Tsu}. We observe that we do not make use of correlation estimates in none of our results.

\end{abstract}

\maketitle
\section{Introduction}
One of the intriguing questions in Statistical Physics is related to the understanding of  how local microscopic perturbations of the dynamics of a particle system, carries through its macroscopic description. In recent years, several articles have been dedicated to the understanding of adding a slow bond, a slow site or a slow boundary to the most classical interacting particle system, namely, \emph{the exclusion process}. For references on this topic, we refer the reader to \cite{patriciaantigo, FGS, bmns} and references therein,  where the hydrodynamic limit for the symmetric simple exclusion process (SSEP)  with, respectively,  a slow bond, a slow site and a slow boundary was analyzed. Recently, the case of  the non-simple symmetric exclusion process with slow boundary has been analyzed in \cite{G1, BGJO, BGJO2} and the asymmetric case in \cite{SS}. In the studied cases mentioned above with a slow boundary, the macroscopic PDE, ends up with boundary conditions of the type: Dirichlet, (linear) Robin, or Neumann.\\

In this article, motivated by deriving other types of boundary conditions, we consider the SSEP in the discrete box $ \{1,\dots,N-1\} $ coupled with slow reservoirs, placed at $x=0$ and $x=N$, whose role is to  inject and  and remove particles in a window of a fixed size $ K\geq 1 $. A particle may enter to the first free site and leave from the first occupied site in its respective window (\textit{i.e.,} $ \{1,\dots K\}, \{N-K,\dots,N-1\} $). We control the action of the reservoirs by fixing the rates of injection/removal as proportional to $ N^{-\theta} $. In {this} article, we address here the characterization of the hydrodynamic and hydrostatic behavior for the slowed regime $ \theta\geq1 $, {and we will consider in the second part of this article \cite{EGN2} the case where $\theta\in (0,1)$, which requires a different set of tools. More precisely, we show that the spatial density of particles is given by a weak solution of the heat equation with non-linear (resp. linear) Robin boundary conditions, if $ \theta=1 $ and $K \geq 2$ (resp. $K=1$, in which we recover the results of \cite{bmns}), and Neumann boundary conditions, if $ \theta>1$ for any $K\geq 1$. For the case $ \theta=1 $, the irreversibility of the boundary dynamics reflects on a non-linear macroscopic boundary evolution for $K\geq 2$ and a simplified version of this model was studied in \cite{DMP12}. The model where particles may enter only through the right and leave only through the left with rates $ \tfrac{1}{2} $ was first introduced by De Masi \textit{et al} in \cite{DMP12}, and the reservoirs were termed "current reservoirs", since they do not fix the value of the density at the boundary, but its gradient. The dynamics we consider here is a generalization of the dynamics of \cite{DMP12} since we allow injection and removal from both reservoirs and moreover, the rate is slowed with respect to the bulk dynamics.  In \cite{DMP12}, the dynamics was shown to have the Propagation of Chaos property, and that result was obtained  by providing sharp estimates on the $L^\infty$ norm of $ v-$functions. As a consequence,  the Fick's Law was shown to hold and the hydrostatic limit was proved in \cite{dptv} and \cite{dptv3}, respectively.\\

When $ K=1 $, we are reduced to the SSEP with classical slowed reservoirs, where the hydrodynamic and hydrostatic scenario were both investigated in \cite{bmns} for $ \theta\geq 0 $ and  for  $\theta<0$, the hydrodynamic behavior was studied in \cite{G1}. For $ \theta\geq0 $, in \cite{bmns}, Baldasso \textit{et al} showed the hydrodynamic limit by the application of the Entropy method, first presented in \cite{GPV}. In their case, which corresponds here to the  case when $K=1$, they were able to use an auxiliary measure which is product and given by a suitable profile and for that reason, the entropy production at the boundaries is small enough to enable them to show a replacement lemma at the boundaries. In the present paper, we apply a similar strategy for $ \theta\geq1 $, but with an extra difficulty due to the explicit correlation terms at the boundaries, which makes us use another replacement lemma. {The strategy works for $\theta\geq 1$, since, in this regime,  the reservoirs' action is sufficiently slow, and we are allowed to use an auxiliary measure of product type which is close to the stationary state of the system. Unfortunately, that same procedure is not possible for $ \theta<1 $ since the comparison measure is quite far from being of product type. Due to the boundary terms of the dynamics, we are not able to control the errors coming from the comparison between Dirichlet form and the carr\'e du champ operator and, as a consequence,  we cannot apply the entropy method, except in the case where the boundary is quite slow, namely $\theta\geq 1$.}  For this reason, in the second part of this article \cite{EGN2}, we make use of duality estimates { obtained  in \cite{E18, ELX18} for the case $\theta=0$} to derive both the hydrostatic and hydrodynamic limit in the case $\theta\in(0,1)$.\\

As a consequence of the hydrodynamic limit we derive Fick's Law. More precisely, we consider two currents related to the system, the conservative and the non-conservative. The former counts the net number of particles going through a bond, while the later counts the  number of particles injected minus the number of particles removed from the system through a site. Then, we associate the corresponding fields and we show their convergence. This is the content of Theorem \ref{thm:current} whose proof is given in Section \ref{sec:currrent}. \\

 Having the hydrodynamic limit proved, it is simple to obtain the hydrostatic limit, by showing that the stationary correlations of the system vanish as the system size grows to infinity. When $K=1 $ that is exactly the strategy pursued in \cite{bmns}. In our case, when $K\geq 2$ we do not have any information about the stationary correlations of the system and for that reason we have to do it in a different way.  Therefore, here the hydrostatic behavior is investigated through the methods developed in \cite{Tsu} and \cite{LT}. In particular, we will follow essentially \cite{Tsu}, where the hydrostatic limit was shown for $ K=1 $. The proof presented in \cite{LT} is robust enough for the hydrostatic limit to follow directly from the hydrodynamic limit when $ \theta=1 $, thus we will focus on the case $ \theta>1 $ and refer the interested reader to \cite{LT} and references therein. Our main interest is when $ \theta>1 $, where the macroscopic evolution is governed by a Neumann Laplacian on $ [0,1] $. In contrast to the arguments in \cite{bmns}, where the hydrostatic limit was shown through estimates on the density and correlation fields, the method in \cite{Tsu} is based on the study of the system's evolution at a subdiffusive time scale. This allows us to show replacement lemmas that, under a different time scale, do not hold. In this sense, our results regarding the hydrostatic limit also extend the ones obtained in \cite{bmns} for $ \theta\geq1 $ by the application of a simpler method and when correlation estimates are not easy to obtain.\\

 Regarding the results of the present paper, as already mentioned, the model expresses a macroscopic phase transition from non-linear Robin to Neumann boundary conditions. In particular, we derive the following hydrodynamic equation when $ \theta=1 $ 
\begin{equation}\label{cauchyprob}
\begin{cases}
&\partial_{t}\rho_{t}(u)=  \partial^2_u\, {\rho} _{t}(u), \quad (t,u) \in [0,T]\times(0,1),\\
&\partial_{u}\rho _{t}(0)= -D_{\alpha,\gamma}\rho_t(0)
,\quad t \in (0,T],\\
&\partial_{u} \rho_{t}(1)=D_{\beta,\delta}\rho_t(1)
,\quad t \in (0,T], \\
&\rho(0,\cdot)= f_0(\cdot),
\end{cases}
\end{equation}
where $ \alpha=(\alpha_1,\dots,\alpha_K), \beta=(\beta_1,\dots,\beta_K), \delta=(\delta_1,\dots,\delta_K), \gamma=(\gamma_1,\dots,\gamma_K) $ are parameters of the boundary dynamics and the operator $ D_{\lambda,\sigma} $ is defined for any vectors $ \lambda=(\lambda_1,\dots,\lambda_K),\sigma= (\sigma_1,\dots,\sigma_K)$  and $ f:[0,1]\to\mathbb{R} $ as
\begin{align*}
	(D_{\lambda,\sigma}f)(u)=
	\sum_{x=1}^K 
	\{
	\lambda_x(1-f(u))f^{x-1}(u)-\sigma_xf(u)(1-f(u))^{x-1}
	\}.
\end{align*}
{In the case $ \theta>1 $, the non linear Robin boundary conditions are replaced with Neumann boundary conditions }
\begin{equation}\label{cauchyprob2}
{\begin{cases}
&\partial_{t}\rho_{t}(u)=  \partial^2_u\, {\rho} _{t}(u), \quad (t,u) \in [0,T]\times(0,1),\\
&\partial_{u}\rho _{t}(0)=\partial_{u}\rho _{t}(1)=0,\quad t \in (0,T], \\
&\rho(0,\cdot)= f_0(\cdot){.}
\end{cases}}
\end{equation}
We also prove uniqueness of the weak solution of \eqref{cauchyprob} in the case where the parameters satisfy {suitable conditions (cf. \eqref{H0} below)}.\\

{The proof of our results is simpler than the one of  \cite{DMP12} and does not require any knowledge on the decay of $v$-functions.  Of course that the estimate on $v$-functions obtained in \cite{DMP12}  can have other purposes than just the hydrodynamic limit, as, for example, the density fluctuations of the system, but in what concerns the hydrodynamics, our  proof is simple and it  relies on good estimates between the Dirichlet form and the carr\'e du champ operator and a few replacement lemmas which allow to control boundary terms. }
 Throughout the paper we will state the results for $ K\geq 2 $, but in some cases present  the proofs in detail for $ K=2 $ only, since for $ K>2 $ the techniques are exactly the same and the biggest change is in the notation. Nevertheless, whenever required, we will state some appropriate remarks regarding the general case $K>2$. For  $ \beta_x=\gamma_x=1 $ and $ \delta_x=\alpha_x=0 $ for all $x\in \{1,\cdots, K\}$, the uniqueness for the Cauchy problem \eqref{cauchyprob} was shown in  \cite{dptv}. For $ K=2 $ with $ \alpha_2=\gamma_2 $ and $ \beta_2=\delta_2 $ the proof reduces to the case of linear Robin boundary conditions, whose uniqueness problem was  studied in \cite{bmns}.\\

{Since we treat in \cite{EGN2} the case $\theta\in(0,1)$, the main issue left open is related to the fluctuations around the hydrodynamic limit,} for which we need to obtain very sharp estimates on the space-time correlations of the system. Large deviations from the stationary state is also another challenge to look at in the near future. Note that in order to get exact information about the stationary state of the system, we cannot make use of  the preliminary work on the matrix product ansatz of Derrida \cite{Derrida}, since it does not straightforwardly apply to this dynamics in general, and encompasses the case $ K=1 $ only.\\

The article is divided as follows. In section \ref{sec:model} we present the model, the notation, the weak formulation for the solution of the Cauchy problem and the main results, namely, the hydrodynamic limit (Theorem \ref{th:hyd_ssep}), {a law of large numbers for the current (Theorem \ref{thm:current}), and the} hydrostatic limit (Theorem \ref{thm:hydrostatics}). In section \ref{sec:hyd} we show the hydrodynamic limit: we start presenting an heuristic proof for finding the notion of weak solution of the PDEs, we identify the main difficulties in the proof and we present the tools to solve them. Then we proceed with the entropy method: in Proposition \ref{prop:tight} we show tightness of the sequence of empirical measures, which shows that there exists convergent subsequences. With the assumption on the uniqueness of the solution of \eqref{cauchyprob}, we proceed with the characterization of limit points. In particular, in Proposition \ref{prop:charac} we show that the spatial density of particles converges to the solution of \eqref{cauchyprob}.
Section \ref{sec:currrent} is devoted to the proof of  the law of large numbers for the current fields associated to the system. 
 Section \ref{sec:hs} is devoted to the proof of the hydrostatic limit. In the Robin case ($\theta=1$), we require the existence of a unique stationary solution, to which the hydrodynamic solution converges. These two elements are obtained in Sections \ref{sec:statio_uniq} and \ref{appendix:convergence}, respectively. In the Neumann case ($\theta>1$), however, any constant profile is stationary, so that we need one further argument. We therefore show in Section \ref{sec:proofhydrostatics} that the total mass of the system evolves in the subdiffusive time scale $N^{1+\theta}$, and on this time scale it converges to a unique constant which determines the stationary profile.
In the appendix, we prove some technical results required throughout the proofs, namely the replacement lemmas (Appendix \ref{appendix:replacement}), an energy estimate (Appendix \ref{appendix:energy}), and the uniqueness of the weak solution to \eqref{cauchyprob} (Appendix \ref{ap:uni_weak}).

%

 \section{{Model and results}}\label{sec:model}

  \subsection{{The microscopic model}}
  
Denote by $N$ a scaling parameter, which will be taken to infinity later on. For $N\geq{2}$ {we call bulk the discrete set of points   $\Lambda_N:=\{1, \ldots, N-1\}$}.  The exclusion process in contact with stochastic reservoirs is a Markov process, that we denote by $\{\eta_t:\,t\geq{0}\}$, {whose state space is} $\Omega_N:=\{0,1\}^{\Lambda_N}$. 
 The configurations of the state space $\Omega_N$  are denoted by $\eta$, so that for $x\in\Lambda_N$,  $\eta(x)=0$ means that the site $x$ is vacant while $\eta(x)=1$ means that the site $x$ is occupied.
For {any} fixed $ K\in\mathbb{N}^+ $, we define $ I_-^K:=\{1,\dots,K\} $, $ I_+^K:=\{N-K,\dots,N-1\}$.
{We introduce the infinitesimal generator}
\begin{equation}\label{generator_ssep}
\mc L_{N}=\mc L_{N,0}+\tfrac{1}{N^\theta}\mc L_{N,b}
\end{equation}
{acting} on functions $f:\Omega_N\rightarrow \bb{R}$ by  
\begin{equation*}
(\mc L_{N,0}f)(\eta)=
\sum_{x=1}^{N-2}\Big(f(\eta^{x,x+1})-f(\eta)\Big)
\quad\text{and}\quad 
(\mc L_{N, b}f)(\eta)\;=\;(\mc L_{N, -}f)(\eta)+(\mc L_{N, +}f)(\eta)
\end{equation*}
where 
\begin{align}\label{eq:gen_bound_lin}
	(\mc L^{}_{N,\pm}f)(\eta)=
	\sum_{x\in I_\pm^K}
	c_x^\pm(\eta){\Big(}f(\eta^{x})-f(\eta){\Big)}
\end{align} 
and for $x\in I^K_{\pm}\setminus\{1,N-1\}$
\begin{align}\label{rates:boundary}
\begin{split}
c_x^-(\eta)&=\alpha_x\eta(1)\cdots\eta(x-1)(1-\eta(x))+\gamma_x(1-\eta(1))\cdots(1-\eta(x-1))\eta(x),\\
c_x^+(\eta)&=\beta_{N-x}(1-\eta(x))\eta(x+1)\cdots\eta(N-1)+
\delta_{N-x}\eta(x)(1-\eta(x+1))\cdots(1-\eta(N-1))
\end{split}
\end{align}
and $c^-_1(\eta)=\alpha_1(1-\eta(1))+\gamma_1 \eta(1)$ and $c^+_{N-1}(\eta)=\beta_{N-1}(1-\eta(N-1))+\delta_{N-1}\eta(N-1)$.
To simplify notation, we will identify $ \beta_x\equiv\beta_{N-x},\delta_x\equiv\delta_{N-x} $.
{In the formulae above, we shortened}
\begin{equation}\label{tranformations}
\eta^{x,y}(z) = 
\begin{cases}
\eta(z), \; z \ne x,y\\
\eta(y), \; z=x\\
\eta(x), \; z=y
\end{cases}
, \quad \eta^x(z)= 
\begin{cases}
\eta(z), \; z \ne x,\\
1-\eta(x), \; z=x
\end{cases},
\end{equation}
and the $\alpha_i,$ $\gamma_i$, $\beta_i,$ $\delta_i$, for $i=1,\dots,K$ are fixed non-negative constants.  The size $K$ of the boundary is considered to be a fixed constant as well. In other words, as illustrated in Figure \ref{fig:BD}, we consider a stirring dynamics in the bulk, and at the two boundary sets  $I_\pm^K $, particles get created (resp. removed) at the empty (resp. occupied) site closest to the boundary.
 {The role of the parameter $\theta$ appearing in \eqref{generator_ssep}} is to {slow down ($\theta\geq 0$) or speed up $(\theta \leq 0)$ the boundary dynamics relatively to the bulk dynamics}. In this article we restrict ourselves to the case $\theta\geq 1$ and in a companion article \cite{EGN2}, we look at the case  $0< \theta<1$.
{Throughout the article, we therefore fix $\theta\geq 1$ and consider  the Markov process $(\eta_{t})_{t\ge 0} $ with infinitesimal generator  given by $\mathcal L_{N}$. }
 
\tikzset{every picture/.style={line width=0.8pt}} 
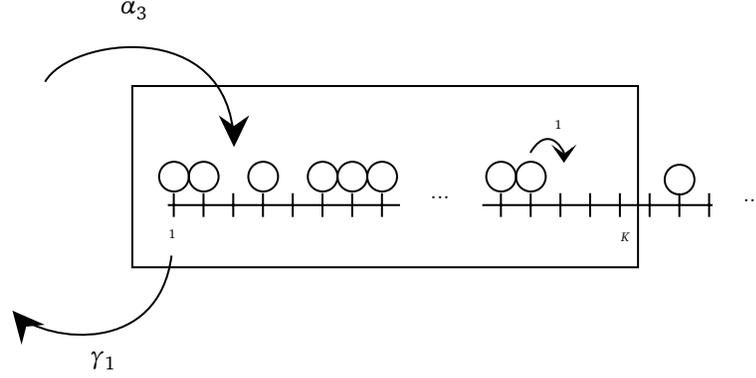
\begin{figure}[H]
	\centering

\tikzset{every picture/.style={line width=0.75pt}} 

\begin{tikzpicture}[x=0.75pt,y=0.75pt,yscale=-1.5,xscale=1.5]

\draw    (311.5,362) -- (109.5,362) (301.5,366) -- (301.5,358)(291.5,366) -- (291.5,358)(281.5,366) -- (281.5,358)(271.5,366) -- (271.5,358)(261.5,366) -- (261.5,358)(251.5,366) -- (251.5,358)(241.5,366) -- (241.5,358)(231.5,366) -- (231.5,358)(221.5,366) -- (221.5,358)(211.5,366) -- (211.5,358)(201.5,366) -- (201.5,358)(191.5,366) -- (191.5,358)(181.5,366) -- (181.5,358)(171.5,366) -- (171.5,358)(161.5,366) -- (161.5,358)(151.5,366) -- (151.5,358)(141.5,366) -- (141.5,358)(131.5,366) -- (131.5,358)(121.5,366) -- (121.5,358)(111.5,366) -- (111.5,358) ;

\draw   (136.57,352.43) .. controls (136.57,349.67) and (138.81,347.43) .. (141.57,347.43) .. controls (144.33,347.43) and (146.57,349.67) .. (146.57,352.43) .. controls (146.57,355.19) and (144.33,357.43) .. (141.57,357.43) .. controls (138.81,357.43) and (136.57,355.19) .. (136.57,352.43) -- cycle ;
\draw   (156.57,352.43) .. controls (156.57,349.67) and (158.81,347.43) .. (161.57,347.43) .. controls (164.33,347.43) and (166.57,349.67) .. (166.57,352.43) .. controls (166.57,355.19) and (164.33,357.43) .. (161.57,357.43) .. controls (158.81,357.43) and (156.57,355.19) .. (156.57,352.43) -- cycle ;
\draw    (68.25,320.5) .. controls (76.62,305.72) and (130.11,298.23) .. (131.71,340.53) ;
\draw [shift={(131.75,342.5)}, rotate = 270] [fill={rgb, 255:red, 0; green, 0; blue, 0 }  ][line width=0.75]  [draw opacity=0] (10.72,-5.15) -- (0,0) -- (10.72,5.15) -- (7.12,0) -- cycle    ;

\draw   (166.57,352.43) .. controls (166.57,349.67) and (168.81,347.43) .. (171.57,347.43) .. controls (174.33,347.43) and (176.57,349.67) .. (176.57,352.43) .. controls (176.57,355.19) and (174.33,357.43) .. (171.57,357.43) .. controls (168.81,357.43) and (166.57,355.19) .. (166.57,352.43) -- cycle ;
\draw   (106.57,352.43) .. controls (106.57,349.67) and (108.81,347.43) .. (111.57,347.43) .. controls (114.33,347.43) and (116.57,349.67) .. (116.57,352.43) .. controls (116.57,355.19) and (114.33,357.43) .. (111.57,357.43) .. controls (108.81,357.43) and (106.57,355.19) .. (106.57,352.43) -- cycle ;
\draw   (97.57,321.93) -- (267.6,321.93) -- (267.6,382.93) -- (97.57,382.93) -- cycle ;
\draw   (176.57,352.43) .. controls (176.57,349.67) and (178.81,347.43) .. (181.57,347.43) .. controls (184.33,347.43) and (186.57,349.67) .. (186.57,352.43) .. controls (186.57,355.19) and (184.33,357.43) .. (181.57,357.43) .. controls (178.81,357.43) and (176.57,355.19) .. (176.57,352.43) -- cycle ;
\draw   (216.57,352.43) .. controls (216.57,349.67) and (218.81,347.43) .. (221.57,347.43) .. controls (224.33,347.43) and (226.57,349.67) .. (226.57,352.43) .. controls (226.57,355.19) and (224.33,357.43) .. (221.57,357.43) .. controls (218.81,357.43) and (216.57,355.19) .. (216.57,352.43) -- cycle ;
\draw   (226.57,352.43) .. controls (226.57,349.67) and (228.81,347.43) .. (231.57,347.43) .. controls (234.33,347.43) and (236.57,349.67) .. (236.57,352.43) .. controls (236.57,355.19) and (234.33,357.43) .. (231.57,357.43) .. controls (228.81,357.43) and (226.57,355.19) .. (226.57,352.43) -- cycle ;
\draw    (110.75,379) .. controls (106.45,412.43) and (69.78,408.46) .. (58.65,398.89) ;
\draw [shift={(57.25,397.5)}, rotate = 409.4] [fill={rgb, 255:red, 0; green, 0; blue, 0 }  ][line width=0.75]  [draw opacity=0] (10.72,-5.15) -- (0,0) -- (10.72,5.15) -- (7.12,0) -- cycle    ;

\draw   (116.57,352.43) .. controls (116.57,349.67) and (118.81,347.43) .. (121.57,347.43) .. controls (124.33,347.43) and (126.57,349.67) .. (126.57,352.43) .. controls (126.57,355.19) and (124.33,357.43) .. (121.57,357.43) .. controls (118.81,357.43) and (116.57,355.19) .. (116.57,352.43) -- cycle ;
\draw   (276.57,353.43) .. controls (276.57,350.67) and (278.81,348.43) .. (281.57,348.43) .. controls (284.33,348.43) and (286.57,350.67) .. (286.57,353.43) .. controls (286.57,356.19) and (284.33,358.43) .. (281.57,358.43) .. controls (278.81,358.43) and (276.57,356.19) .. (276.57,353.43) -- cycle ;
\draw  [fill={rgb, 255:red, 0; green, 0; blue, 0 }  ,fill opacity=1 ] (245.6,343) -- (242.7,347.2) -- (239.8,343) -- (242.7,345.1) -- cycle ;
\draw    (231.4,344.4) .. controls (235.2,337.8) and (240,338.6) .. (242.8,344.6) ;

\draw (97,296.5) node [scale=1] [align=left] {$\displaystyle \ \ \alpha _{3} \ $};
\draw (86.67,414.17) node [scale=1] [align=left] {$\displaystyle \ \ \gamma _{1} \ $};
\draw  [color={rgb, 255:red, 255; green, 255; blue, 255 }  ,draw opacity=1 ][fill={rgb, 255:red, 255; green, 255; blue, 255 }  ,fill opacity=1 ]  (187.9,351) -- (214.9,351) -- (214.9,369) -- (187.9,369) -- cycle  ;
\draw (201.4,360) node [scale=0.7]  {$\cdots $};
\draw (263.5,372.67) node [scale=0.5]  {$K$};
\draw (111,371.67) node [scale=0.5]  {$1$};
\draw (240.8,334.87) node [scale=0.5]  {$1$};
\draw  [color={rgb, 255:red, 255; green, 255; blue, 255 }  ,draw opacity=1 ][fill={rgb, 255:red, 255; green, 255; blue, 255 }  ,fill opacity=1 ]  (293,352) -- (320,352) -- (320,370) -- (293,370) -- cycle  ;
\draw (306.5,361) node [scale=0.7]  {$\cdots $};

\end{tikzpicture}
\caption{Left boundary dynamics.}
\label{fig:BD}
\end{figure}

\subsection{Hydrodynamic equation {and uniqueness}} \label{sec:hyd_eq_ssep}
We now define the macroscopic limit of our model and its topological setup.  We denote by $\langle \cdot,\cdot\rangle _{\mu}$ the inner product  in $L^{2}([0,1])$ with respect to a measure $\mu$ defined in $[0,1]$ and $\| \cdot\|_{L^2 (\mu) }$ is the corresponding norm. When $\mu$ is the Lebesgue measure we write $\langle \cdot,\cdot\rangle$ and $\| \cdot\|_{L^2}$ for the corresponding norm. 

{Fix once and for all a finite time horizon $T>0$}. We denote by $C^{m,n}([0, T] \times [0,1])$ the set of functions defined on $[0, T] \times [0,1] $ that are $m$ times differentiable on the first variable and $n$ times differentiable  on the second variable, {with} continuous derivatives. For a function $G:=G(s,u)\in C^{m,n}([0, T] \times [0,1])$ we denote by $\partial _{s}G$  its derivative with respect to the time variable $s$ { and by $\partial_{u}G$  its derivative with respect to the space variable  $u$.}

Now we want to define the space where the solutions of the hydrodynamic  equations will live on, namely the Sobolev space $\mathcal H^1$ on $[0,1]$. For that purpose, we define the semi inner-product $\langle \cdot, \cdot \rangle_{1}$ on the set $C^{\infty} ([0,1])$ by $\langle G, H \rangle_{1} =\langle \partial_u  G \,, \partial_u  H \rangle$ 
and  the corresponding semi-norm is denoted by $\| \cdot \|_{1}$. 

\begin{definition}
\label{Def. Sobolev space}
The Sobolev space $\mathcal{H}^{1}$ on $[0,1]$ is the Hilbert space defined as the completion of $C^\infty ([0,1])$ for the norm 
$\| \cdot\|_{{\mc H}^1}^2 :=  \| \cdot \|_{L^2}^2  +  \| \cdot \|^2_{1}.$
Its {elements} coincide a.e. with continuous functions. The space $L^{2}(0,T;\mathcal{H}^{1})$ is the set of measurable functions $f:[0,T]\rightarrow  \mathcal{H}^{1}$ such that 
$\int^{T}_{0} \Vert f_{s} \Vert^{2}_{\mathcal{H}^{1}}ds< \infty. $
\end{definition}

We can now give the definition of the weak solution of the  hydrodynamic equation that will be derived for the process described above when $\theta\geq 1$. {Recall that the operator $ D_{\lambda,\sigma} $ is defined for any vectors $ \lambda=(\lambda_1,\dots,\lambda_K),\sigma= (\sigma_1,\dots,\sigma_K)$  and $ f:[0,1]\to\mathbb{R} $ as
\begin{align}\label{exp:D}
	(D_{\lambda,\sigma}f)(u)=
	\sum_{x=1}^K 
	\{
	\lambda_x(1-f(u))f^{x-1}(u)-\sigma_xf(u)(1-f(u))^{x-1}
	\}.
\end{align}}
\begin{definition}
	\label{def:weak_sol_ Robin}
	{Let $f_0:[0,1]\rightarrow [0,1]$ be a measurable function.} We say that  $\rho:[0,T]\times[0,1] \to [0,1]$ is a weak solution of the heat equation with Robin boundary conditions {(this will be obtained in the case $\theta=1$)}
	\begin{equation}\label{eq:Robin_equation}
	\begin{cases}
	&\partial_{t}\rho_{t}(u)=  \partial^2_u\, {\rho} _{t}(u), \quad (t,u) \in [0,T]\times(0,1),\\
	&\partial_u\rho _{t}(0)= -D_{\alpha,\gamma}\rho_t(0)
,\quad t \in [0,T]{,}\\ &\partial_u \rho_{t}(1)=D_{\beta,\delta}\rho_t(1)
	,\quad t \in [0,T] {,}\\
	&\rho(0,\cdot)= f_0(\cdot),
	\end{cases}
	\end{equation}
	if the following two conditions hold: 
	\begin{enumerate}[1.]
		\item $\rho \in L^{2}(0,T;\mathcal{H}^{1})$, 
		
		\item $\rho$ satisfies the weak formulation:
		\begin{equation}\label{eq:Robin_integral}
		\begin{split}
		&F(\rho, G,t)
		:=\langle \rho_{t},  G_{t}\rangle  -\langle f_0 , G_{0} \rangle - \int_0^t\langle \rho_{s},\Big( \partial^2_u + \partial_s\Big) G_{s} \, \rangle {ds}  \\
		&+ \int^{t}_{0}  \Big \{\rho_{s}(1) \partial_u G_{s}(1)-\rho_{s}(0)  \partial_u G_{s}(0) \Big\} \, ds
		 -\int^{t}_{0}G_{s}(1)(D_{\beta,\delta}\rho_s)(1)
		ds
		-\int_0^t G_{s}(0) (D_{\alpha,\gamma}\rho_s)(0) ds=0,
		\end{split}   
		\end{equation}
	for all $t\in [0,T]$, any function $G \in C^{1,2} ([0,T]\times[0,1])$. 
	\end{enumerate}
	
	\medskip
	
We say that  $\rho:[0,T]\times[0,1] \to [0,1]$ is a weak solution of the heat equation with Neumann boundary conditions (this will be obtained in the case $\theta>1$)
\begin{equation}\label{eq:Neumann_equation}
	\begin{cases}
	&\partial_{t}\rho_{t}(u)=  \partial^2_u\, {\rho} _{t}(u), \quad (t,u) \in [0,T]\times(0,1),\\
	&\partial_u\rho _{t}(0)=\partial_u\rho _{t}(1)=0,\quad t \in [0,T]{,} \\
	&\rho(0,\cdot)= f_0(\cdot),
	\end{cases}
	\end{equation}
if conditions 1. and 2. above hold, with \eqref{eq:Robin_integral} replaced by  
\begin{equation}
\label{eq:Neumann_integral}
F(\rho, G,t):=\langle \rho_{t},  G_{t}\rangle  -\langle f_0 , G_{0} \rangle - \int_0^t\langle \rho_{s},\Big( \partial^2_u + \partial_s\Big) G_{s} \, ds\rangle+ \int^{t}_{0}  \Big \{\rho_{s}(1) \partial_u G_{s}(1)-\rho_{s}(0)  \partial_u G_{s}(0) \Big\} \, ds=0.
\end{equation}
\end{definition}

{\begin{rem}
Observe that  since $\rho \in L^{2}(0,T;\mathcal{H}^{1})$, above in \eqref{eq:Robin_integral} the quantities $\rho_s (0)$ and $\rho_s (1)$ are well defined  for almost every time $s$.
\end{rem}}
Throughout the present article we make the following assumption in the case $\theta=1$ (Robin {boundary} conditions),
	\begin{equation}
	\label{H0} 
	\tag{H0} 
	\mbox{The (finite) sequences } \;\; \alpha, \;\;  \gamma,\;\; \beta \;\; {\textrm{and}}\;\; \delta\mbox{ are non-increasing},
	\end{equation}
which ensures uniqueness of the weak solutions of equation \eqref{eq:Robin_equation}:
\begin{lem}\label{lem:uniqueness}[Uniqueness of weak solutions] Consider the notion of weak solution introduced in Definition \ref{def:weak_sol_ Robin}, and fix a measurable initial profile $f_0:[0,1]\to [0,1]$. 
Assuming \eqref{H0},
the weak solution of \eqref{eq:Robin_equation} is unique. Moreover,  the weak solution of \eqref{eq:Neumann_equation} is unique. 
\end{lem}
The proof of the first statement is postponed to Appendix \ref{ap:uni_weak}. The Neumann case is classical and for that reason it is omitted, but the proof can be found in \cite{patriciaantigo}. For the sake of concision, we do not recall for each of our main results that assumption \eqref{H0} is made, however since it guarantees uniqueness of weak solutions, this assumption is made throughout the article  whenever $\theta=1$.
\begin{rem}
For $ \beta=\gamma\equiv j $ and $ \delta=\alpha\equiv0 $  we recover the boundary conditions of \cite{DMP12}.	
\end{rem}
\begin{rem}
For $ K=2 $, \eqref{eq:Robin_equation} rewrites as
 \begin{equation}\label{eq:Robin_equation_new_new}
 \begin{cases}
 &\partial_{t}\rho_{t}(u)=  \partial^2_u\, {\rho} _{t}(u), \quad (t,u) \in [0,T]\times(0,1),\\
 &\partial_u\rho _{t}(1)= \beta_1-(\beta_1+\delta_1)\rho_{t}(1)+(\delta_2-\beta_2)(\rho_t^2(1)-\rho_t(1)),\quad t \in [0,T] {,}\\ &\partial_u\rho_{t}(0)=\rho_t(0)(\alpha_1+\gamma_1)-\alpha_1-(\gamma_2-\alpha_2)(\rho_{t}^2(0)-\rho_t(0)),\quad t \in [0,T] {,}\\
 &\rho(0,\cdot)= f_0(\cdot),
 \end{cases}
 \end{equation}
and for $\alpha_2=\gamma_2$ and $\beta_2=\delta_2$ we recover the linear Robin boundary conditions as in \cite{bmns} and when  $\alpha_2=\gamma_2=\beta_2=\delta_2=0$ and $\beta_1=1-\delta_1=\beta$ and   $\alpha_1=1-\alpha_1=\alpha$ we deal with exactly the same model of \cite{bmns} and we recover their result.
{Further note that the weak solution of} \eqref{eq:Robin_equation_new_new} when $\alpha_2=\gamma_2$ and $\beta_2=\delta_2$ (corresponding to linear Robin boundary conditions) is shown in \cite{bmns} to be unique. 
\end{rem}

\subsection{Hydrodynamic limit}
\label{sec:HL}

In this section we state the hydrodynamic limit of the process $\{\eta_{{tN^2}}\}_{t\geq{0}}$. Note the scaling factor $N^2$ whose purpose is to accelerate the process to a diffusive time scale. Let ${\mc M}^+$ be the space of positive measures on $[0,1]$ with total mass bounded by $1$ equipped with the weak topology. For any configuration  $\eta \in \Omega_{N}$ we define the empirical measure ${\pi^{N}(\eta,\cdot)\in{\mc M}^+}$ on $[0,1]$ as
\begin{equation}\label{MedEmp}
\pi^{N}(\eta, du)=\dfrac{1}{N-1}\sum _{x\in \Lambda_{N}}\eta(x)\delta_{\frac{x}{N}}\left( du\right),
 \end{equation}
where $\delta_{a}$ is a Dirac mass on $a \in [0,1]${. Given the trajectory $\{\eta_{{tN^2}}\}_{t\geq{0}}$ of the \emph{accelerated} process, we further introduce 
$\pi^{N}_{t}(du):=\pi^{N}(\eta_{tN^2}, du)$ the empirical measure at the macroscopic time $t$.} 
Below, and in what follows, we use the notation $\langle \pi^{{N}}_t, G\rangle$ to denote the integral of $G$ w.r.t. the measure $\pi_t^{{N}}$. This notation should not be confused with the inner product in $L^2([0,1])$.
Fix $T>0$ and $\theta\geq 0$. We denote by $\PP _{\mu _{N}}$ the probability measure in the Skorohod space $\mathcal D([0,T], \Omega_N)$ induced by the  Markov process $\{\eta_{{tN^2}}\}_{t\geq{0}}$ and the initial probability measure $\mu_N$ and  $\EE _{\mu _{N}}$ denotes the expectation w.r.t. $\PP_{\mu _{N}}$.  

\begin{definition}\label{def:meas_ass}
{We say that} a sequence of probability measures $\lbrace\mu_{N}\rbrace_{N\geq 1 }$ on $\Omega_{N}$  is associated with a profile {$\rho_{0}:[0,1]\rightarrow[0,1]$} if for any continuous function $G:[0,1]\rightarrow \mathbb{R}$  and every $\delta > 0$ 
\begin{equation}\label{assoc_mea}
  \lim _{N\to\infty } \mu _{N}\Big( \eta \in \Omega_{N} : \big|\langle {\pi^N}, G\rangle - \langle G,\rho_{0}\rangle\big|    > \delta \Big)= 0.
\end{equation}
\end{definition}

Our first result is the hydrodynamic limit for the process introduced above and it is stated as follows. 

\begin{thm}
\label{th:hyd_ssep}
Let {$f_0:[0,1]\rightarrow[0,1]$} be a measurable function and let $\lbrace\mu _{N}\rbrace_{N\geq 1}$ be a sequence of probability measures in $\Omega_{N}$ associated with {$f_0$} {in the sense of Definition \ref{def:meas_ass}}. Then, for any $t\in[0,T]$ {and every $\delta>0$},
\begin{equation*}\label{limHidreform}
 \lim _{N\to\infty } \PP_{\mu _{N}}\Big( \big| \langle \pi^N_t, G\rangle - \langle G,\rho_{t}\rangle\big|   > \delta \Big)= 0,
\end{equation*}
where  $\rho_{t}(\cdot)$ is the unique weak solution, in the sense of Definition \ref{def:weak_sol_ Robin}, of \eqref{eq:Robin_equation} for $\theta=1$, resp. \eqref{eq:Neumann_equation} for $\theta>1$.
\end{thm}
{Let $\mathcal D([0,T],\mathcal{M}^{+})$ be the Skorohod space of trajectories in $\mathcal M^+$.}
{Let $\lbrace\mathbb{Q}_{N}\rbrace_{N\geq 1}$ be the  sequence of probability measures on $\mathcal D([0,T],\mathcal{M}^{+})$ induced by the  Markov process $\{\pi_{t}^N\}_{t\geq{0}}$  and $\mathbb{P}_{\mu_{N}}$, {namely} ${\mathbb{Q}_{N}=\mathbb{P}_{\mu_{N}}\circ (\pi^N)^{-1}}$.
To prove Theorem \ref{th:hyd_ssep} we first show that the sequence  $\lbrace\mathbb{Q}_{N}\rbrace_{N\geq 1}$ is tight, and then prove that any of its limit points $\mathbb Q$ is concentrated on trajectories of  measures that are  absolutely continuous with respect to the Lebesgue measure  (this is a consequence of the exclusion dynamics), whose density $\rho_{t}(u)$ is the unique (cf. Lemma \ref{lem:uniqueness}) weak solution of  the hydrodynamic equation.} 
We prove Theorem \ref{th:hyd_ssep} in Section \ref{sec:hyd}.

\subsection{Empirical currents}
\label{sec:currentsthm}
{Our next result is a law of large numbers for the empirical currents of the process}. Let $J^N_t(x)$ denote the process that counts the flux of particles {(in the accelerated process $(\eta_{sN^2})_{s\geq 0}$)} through the bond $\{x,x+1\}$ up to time ${t}$, i.e. the number of  particles that  jumped from the site $x$ to the site $x+1$ minus the number of  particles that  jumped from the site $x+1$ to the site $x$ during the time interval $[0,{t}]$. The empirical measure associated with this  \textit{conservative current} is defined as 
{\begin{equation*}
J^N_t(du)\;:=\; \frac{1}{N^2}\sum_{x=1}^{N-2}J^N_t(x) \delta_{\frac xN}(du)\,.
\end{equation*}}
Notice the normalization factor $N^2$ which is taking into account the diffusive time rescaling and the space normalization. 
For $x\in {I^K_\pm}$, we  denote by $K^N_t(x)$ the non-conservative current at the site $x$ up to time ${t}$, that is,  the number of particles that have been created minus the number of particles that have been {removed from} the system {at} site $x$. The corresponding empirical measure is given by 
{\begin{equation*}
 K^N_t(du) \;:=\frac{1}{N}\sum_{x\in {I_+^K\cup I_-^K}}\; K^N_t(x) \delta_{\frac xN}(du). 
\end{equation*}} For a test function $f$ we use the notation $\langle J^N_t , f\rangle$ and $\langle K^N_t , f\rangle$ to denote, respectively:
\begin{equation*}
 \langle J^N_t, f\rangle \;:=\frac{1}{N^2}\sum_{x=1}^{N-2}J^N_t(x)f(\tfrac xN)
\quad\textrm{and}\quad
 \langle K^N_t, f\rangle \;:=\frac{1}{N}\sum_{x\in {I_+^K\cup I_-^K}}\; K^N_t(x)f(\tfrac xN). 
\end{equation*}
Our second main result is a law of large numbers for the current fields.

\begin{thm}[Law of large Numbers for the current]
\label{thm:current}
For any $t\in[0,T]$, {$f\in C^1([0,1])$ and every $\delta>0$},
\begin{align*}
&\lim_{ N \rightarrow +\infty }
\bb P_{\mu_{N}} \Bigg[ \Big\vert \langle J^N_t, f\rangle - \int_0^t\int_0^1 f(u)\,\partial_u  \rho_s(u)\, du\,ds \Big\vert
> \delta \Bigg] \;=\; 0\,,\\
&\lim_{ N \rightarrow +\infty }
\bb P_{\mu_{N}} \Bigg[ \Big\vert \langle K^N_t, f\rangle - \textbf{1}_{\{\theta=1\}}\int_0^t f(0)(D_{\alpha,\gamma}\rho_s)(0)+f(1)(D_{\beta,\delta}\rho_s)(1)\,ds
 \Big\vert
> \delta \Bigg] \;=\; 0\,,
\end{align*}
where $\rho_t(\cdot)$ is the unique weak solution of \eqref{eq:Robin_equation} if $\theta=1$ (resp. of \eqref{eq:Neumann_equation} if $\theta>1$). In {particular}, writing $ j_t^N=J_t^N+K_t^N $, we have that $ j^N $ converges weakly to $ jdu $, where $ j $ { is given by} $ j=-\nabla\rho $.
\end{thm}
This theorem is proved in Section \ref{sec:currrent}.

\subsection{Hydrostatic limit}
\textcolor{red}{Let 
$
	\mathfrak{i}_1:=\alpha_1+\beta_1
	$ and $\mathfrak{o}_1:=\gamma_1+\delta_1,$
	and observe that under the conditions $ \mathfrak{i}_1\neq0 $ and $\mathfrak{o}_1\neq 0 $ the Markov process $\eta_t$ is irreducible on its finite state space $\Omega_N $. 
	To see this, it is enough to note that if  $ \mathfrak{i}_1=0 $ (resp. or $\mathfrak{o}_1=0 )$ and by taking a configuration with a single particle at distance one from the boundary (resp. or  a configuration with a single hole at a distance one from the boundary), the  empty (resp. full) configuration is absorbing.
	 For future reference, we introduce
	\begin{align}
	\mathfrak{i}_1\neq 0 \quad \textrm{and}\quad\mathfrak{o}_1\neq 0 \tag{H1}\label{H1}.
	\end{align}
Under \eqref{H1}, we shall denote the unique stationary measure of the process by $ \mu_N^{ss} $.} Our third main result concerns the hydrostatic limit for the dynamics, which gives the macroscopic behavior of our model starting from the stationary state  $ \mu_N^{ss} $. One important ingredient in our proof is the uniqueness of the stationary solution of the hydrodynamic equation. 
In the regime $\theta=1$, on the other hand, to establish the hydrostatic limit, it is sufficient to show that there is a unique stationary solution to the hydrodynamic equation \eqref{eq:Robin_equation}. 
In the  regime  $ \theta>1 $ any constant profile is a stationary solution to \eqref{eq:Neumann_equation}, however under suitable assumptions, this constant can be uniquely determined as the mass $ m^*$ to which the microscopic system relaxes on subdiffusive timescales.

To establish the different uniqueness results above, we will need in the two cases ($\theta=1$ and $\theta>1$) further assumptions. For this reason, we introduce

\begin{align}
	\textcolor{red}{\eqref{H1}\quad \text{and} \quad }
	\delta_1\leq \alpha_1
	,\;
	\beta_1\leq\gamma_1,
	&\quad\text{or}\quad
	\textcolor{red}{\eqref{H1}\quad \text{and} \quad }
		\delta_1\geq \alpha_1
	,\;
		\beta_1\geq\gamma_1,\tag{H2}\label{H2}\\
	\textcolor{red}{\eqref{H1}\quad \text{and} \quad }
	\alpha+\beta 
	,\;
	\gamma+\delta& \quad  \mbox{are non increasing.} \tag{H3}\label{H3}
\end{align}
We are now ready to state our third main result.
\begin{thm}
\label{thm:hydrostatics}
For $\theta=1$, assuming \eqref{H2}  there exists a unique stationary solution $\rho^*$ of \eqref{eq:Robin_equation},  {such that}  $\mu ^{ss}_{N}$  is associated with it in the sense of Definition \ref{def:meas_ass}, i.e. {for every $\delta>0$ and $G \in C([0,1])$}
\begin{equation*} 
\lim _{N\to\infty } \mu ^{ss}_{N}\Big( \big|\langle {\pi^N}, G\rangle - \langle G,\rho^*\rangle\big|    > \delta \Big)= 0.
\end{equation*}	
For $\theta>1$, assuming \eqref{H3} there exists a unique constant $m^*\in [0,1]$, such that $\mu ^{ss}_{N}$ is associated with the constant profile $\rho^*\equiv m^*$.
\end{thm}
\begin{rem}
[On assumptions \eqref{H2} and \eqref{H3}]
Assumption \eqref{H2} is used in the case $ \theta=1 $ to guarantee uniqueness of the stationary solution. Assumption \eqref{H3} is used for $ \theta>1 $ to prove convergence of the mass of the system to a defined constant.
As we will see through the article, one could weaken these assumptions, yet we elected to settle for assumptions \eqref{H2} and \eqref{H3} to provide the reader with a working case, since finding optimal bounds for both of the cases is a non-trivial algebraic problem that goes beyond the scope of this article.
\end{rem}
The proof of Theorem \ref{thm:hydrostatics} is the purpose of Section \ref{sec:hs}.

\section{Proof of Theorem \ref{th:hyd_ssep}}\label{sec:hyd}

In this section we present the proof of the hydrodynamic limit and we start by giving an heuristic argument in order to deduce the notion of weak solution given in Definition \ref{def:weak_sol_ Robin}.  To simplify the exposition we present the proof  for the case $K=2$ but  emphasize that the general case follows straightforwardly. 
\subsection{Heuristic argument}
We start by briefly outlining the argument before detailing the relevant steps of the proof in Section \ref{sec:characterization}.
Let us fix a test function $G\in C^{1,2}([0,T]\times [0,1])$. Following from Dynkin's formula and simple computations, 
\begin{equation}\label{Dynkin'sFormula}
\begin{split}
{M}_{t}^{N}(G)&:= {\langle \pi_{t}^{N},G_t\rangle -\langle \pi_{0}^{N},G_0\rangle-\int_{0}^{t}\langle \pi^N_s,(\partial_s+N^2\mathcal{L}_N)G_s\rangle\, ds}\\
&= \langle \pi_{t}^{N},G_t\rangle -\langle \pi_{0}^{N},G_0\rangle-\int_{0}^{t}\langle \pi^N_s,(\partial_s+\Delta_N)G_s\rangle\, ds\\
&-\int_0^t\nabla _N^+G_s(0)\eta_{sN^2}(1)-\nabla_N^-G_s(1)\eta_{{sN^2}}(N-1)\, ds\\
&-{N^{1-\theta}}\int_0^t
G_s(\tfrac1N)
\{
\alpha_1-\eta_{{sN^2}}(1)(\alpha_1+\gamma_1)
\}
+G_s(\tfrac{N-1}{N})
\{
\beta_1-\eta_{{sN^2}}(N-1)(\beta_1+\delta_1)
\}ds\\
&-{N^{1-\theta}}\int_0^t
G_s(\tfrac2N)
\{
\alpha_2\eta_{{sN^2}}(1)-\gamma_2\eta_{{sN^2}}(2)-\eta_{{sN^2}}(1)\eta_{{sN^2}}(2)(\alpha_2-\gamma_2)
\}ds
\\&-{N^{1-\theta}}\int_0^tG_s(\tfrac{N-2}{N})
\{
\beta_2\eta_{{sN^2}}(N-1)-\delta_2\eta_{{sN^2}}(N-2)-\eta_{{sN^2}}(N-1)\eta_{{sN^2}}(N-2)(\beta_2-\delta_2)
\}ds
\end{split}
\end{equation}
is a martingale with respect to the natural filtration  $\{\mathcal{F}_{t}\}_{ t\ge 0}$, where for each $t\ge 0$, $\mathcal{F}_t:=\sigma({\eta_{sN^2}}: s \leq  t)$. 
{Above, for $x\in\Lambda_N$,  the discrete derivatives  of  $G$ are defined by
$$\nabla^+_{N} G_{s}(\tfrac{x}{N})= N \big[ G (\tfrac{x+1}{N}) -G (\tfrac{x}{N}) \big],$$ $\nabla^-_{N} G(\tfrac{x}{N})= \nabla^+G(\tfrac {x-1}{N})$ and its discrete laplacian is defined by $$\Delta_{N} G(\tfrac{x}{N})= N^2 \big[ G (\tfrac{x+1}{N}) -2G (\tfrac{x}{N})+G (\tfrac{x+1}{N})  \big].$$
}
\begin{rem}
	For fixed $ K\geq2 $, the expression above can be compactly written by introducing the operators $ D^{N,\pm}_{\cdot,\cdot} $ defined by
	\begin{align}\label{def:D^N}
	\begin{split}
		(D_{\lambda,\sigma}^{N,-}f)(x)&=\Big\{{ \lambda_x}f(1)\dots f(x-1) (1-f(x))-{\sigma_x}(1-f(1))\dots (1-f(x-1)) f(x) \Big\}1_{x\in I_-^K}\\
		(D_{\lambda,\sigma}^{N,+}f)(x)&=\Big\{ \lambda_x (1-f(x))f(x+1)\dots f(N-1)-\sigma_x f(x)(1-f(x+1))\dots (1-f(N-1)) \Big\}1_{x\in I_+^K}			
	\end{split}
	\end{align}
for $ f:\mathbb{Z}\to\mathbb{R} $ and $ \lambda=(\lambda_1,\dots,\lambda_K), \sigma=(\sigma_1\dots,\sigma_K) $. {With this notation,} Dynkin's formula takes the form
\begin{equation*}
\begin{split}
{M}_{t}^{N}(G)&= \langle \pi_{t}^{N},G_t\rangle -\langle \pi_{0}^{N},G_0\rangle-\int_{0}^{t}\langle \pi^N_s,(\partial_s+\Delta_N)G_s\rangle\, ds\\
&-\int_0^t
\Big\{
\nabla _N^+G_s(0)\eta_{{sN^2}}(1)-\nabla_N^-G_s(1)\eta_{{sN^2}}(N-1) 
\Big\}
ds\\
&-{N^{2-\theta}}\int_0^t
\Big\{
\inner{\pi^N(D_{\alpha,\gamma}^{N,-}\eta_{{sN^2}},\cdot),G_s}
+\inner{\pi^N(D_{\beta,\delta}^{N,+}\eta_{{sN^2}},\cdot),G_s} 
\Big\}ds.
\end{split}
\end{equation*}
In fact, the {main} technical issue with the proof of the hydrodynamic limit is, as we will also see ahead, the proof of the replacement lemmas which {roughly states}
$ \tfrac{N^2}{N^\theta}\inner{\pi^N(D_{\alpha,\gamma}^{N,-}\eta_{{sN^2}},\cdot),G_s}\xrightarrow{\mathbb{P}_{\mu_N}}(D_{\alpha,\gamma}\rho_s)(0)G_s(0) $ for $ \theta=1 $, with $ \rho(\cdot) $ being the unique weak solution to \eqref{eq:Robin_equation}.
\end{rem}
As we will show, $\mathbb{E}_{\mu_N} \Big[ \big(M_{t}^{N}(G)\big)^2\Big]$ vanishes as $N\to\infty$. We now focus on the  integral terms above. Let us start with the  boundary term coming from the bulk dynamics, that is, the  term on the second line of the previous display. By Theorem \ref{replacement lemma_1}, with the choice $\psi\equiv 1$,    we are able to replace {the time integral of $\eta_{{s}}(1)$ (resp.  $\eta_{{s}}(N-1)$) by the  time integral of the } average in a box  of size $\lfloor\varepsilon N\rfloor$ {to the right of site} $1$ (resp. {to the left of site} $N-1$):
  \begin{equation}\label{eq:emp_aver}
  \overrightarrow{\eta}_{{s}}^{\varepsilon N}(1):=\frac{1}{\varepsilon N}\sum_{x=2}^{1+\varepsilon N}\eta_{{s}}(x), \quad \overleftarrow{\eta}_{{s}}^{\varepsilon N}(N-1):=\frac{1}{\varepsilon N}\sum_{x=N-2}^{N-1-\varepsilon N}\eta_{{s}}(x),\end{equation} then,
since, for $N $ sufficiently big, 
$
\overrightarrow{\eta}^{\varepsilon N}_{{s}}(1)\sim\rho_s(0)$ (resp. $ \overrightarrow{\eta}_{{s}}^{\varepsilon N}(N-1)\sim\rho_s(1))$
 in a sense which will be explained later on, and by a Taylor expansion on  the {test} function $G$, we arrive  at
\begin{equation*}
\begin{split}
\int_0^t\partial_u G_s(0)\rho_s(0)-\partial_u G_s(1)\rho_s(1) \,ds.
\end{split}
\end{equation*}
which is exactly the fourth term at the right hand side of \eqref{eq:Robin_integral}. 
By abuse of notation, above and below $\varepsilon N$ denotes $\lfloor\varepsilon N\rfloor$.
 Now we analyse the terms coming from the boundary dynamics.  We start with the terms on the fourth line on the right hand-side of \eqref{Dynkin'sFormula}. Note that when $\theta>1$, since $G$ and $\eta$ are bounded,  these terms are of order $O(N^{1-\theta})$ and so they vanish as $N\to+\infty$.  When $\theta=1$ and using again Theorem \ref{replacement lemma_1}, with the choice $\psi\equiv 1$, those terms are going to contribute to the integral formulation with 
 \begin{equation*}
\int_0^t G_s(0)(\alpha_1-(\alpha_1+\gamma_1)\rho_{s}(0))+  G_s(1)(\beta_1-(\beta_1+\delta_1)\rho_{s}(1)) \, ds.
\end{equation*}
 Now we look at the fifth and sixth terms at the right hand-side of \eqref{Dynkin'sFormula}. We focus on the terms on the fifth line, but we note that the analysis is completely analogous for the terms in the sixth line. As before, for $\theta>1$ those terms are of order $O(N^{1-\theta})$ and so they vanish as $N\to+\infty$. When $\theta=1$,  from Theorem \ref{replacement lemma_0}, with the choice $\varphi\equiv 1$, we can replace, for any term that does not involve the product of $\eta(1)$ and $\eta(2)$, $\eta(2)$ by $\eta(1)$ and from Theorem \ref{replacement lemma_1} (with  $\psi\equiv 1$), replace $\eta(1)$ by  $ {\eta}^{\varepsilon N}(1)$. For the quadratic terms in $\eta(1)\eta(2)$ we first apply Theorem \ref{replacement lemma_1} (with $\psi(\eta)=\eta(1) $), to replace $\eta(2)$ by ${\eta}^{\varepsilon N}(1)$. In the resulting term $\eta(1)\eta^{\varepsilon N}(1)$, we then replace $\eta(1)$ by ${\eta}^{\varepsilon N}(1)$ by applying Theorem \ref{replacement lemma_1} (with $\psi(\eta)={\eta}^{\varepsilon N}(1)$).
From this we conclude that the terms on the fifth line of \eqref{Dynkin'sFormula} contribute to the integral formulation with 
 \begin{equation*}
\int_0^t G_s(0)(\alpha_2-\gamma_2)(\rho_{s}^2(0)-\rho_{s}(0)) \, ds.
\end{equation*}

Recall that we defined after Theorem \ref{th:hyd_ssep} the distribution $\mathbb{Q}_N$ of the trajectory of the empirical measure $\pi^N$. Assuming that one proves that  the sequence $ \{\mathbb{Q}_N\}_{N\geq 1} $ is tight (which is done in Section \ref{sec:tight}), the arguments above prove that any of its limit points is a Dirac measure supported on the trajectory $\pi_t(du)=\rho_t(u)du$ where $\rho_t(\cdot) $ is the unique weak solution of  \eqref{eq:Robin_equation}. These arguments are carried out in further detail in the next subsections.

\subsection{Tightness}\label{sec:tight}

\begin{prop}\label{prop:tight}
	The sequence $ \{\mathbb{Q}_N\}_{N\geq 1} $ is \textit{tight} under the \textit{Skorohod} topology of $ \mathcal{D}([0,T],\mathcal{M}^+) $.
\end{prop}

\begin{proof}
	From Chapter 4 of \cite{KL}, in order to prove tightness  it is enough to show that 
		\begin{align*}
	\lim_{\gamma \to 0}\limsup_{N\to+\infty}\sup_{\tau\in\mathcal{T}_N,{\lambda}\leq\gamma}\mathbb{P}_{\mu_N}^N\left(\eta_{\cdot}\in\mathcal{D}([0,T],\Omega_N):\abs{\langle\pi^N_{\tau+\lambda},G\rangle-\langle\pi^N_{\tau},G\rangle}>\epsilon	\right)=0,
	\end{align*}
	for any continuous function  $G:[0,1]\to\bb R$ {and every $\epsilon >0$}.  Above $ \mathcal{T}_N $ is the set of stopping times bounded by $ T $.  In  fact, we are going to prove the result for functions in $C^2([0,1])$, but then, by an $L^1$ approximation it is simple to extend the result to continous functions. By Proposition 4.1.7 in \cite{KL} it is enough to show the result for every
function G in a dense subset of C([0; 1]), with respect to the uniform topology.
	From now on we assume that $G\in C^2([0,1])$. From Dynkin's formula, plus 
	 Chebyshev's and Markov's inequality, we can bound the previous probability by
	\begin{align*}\begin{split}
	& \frac{{2}}{\epsilon} \mathbb{E}_{\mu_N}\left[\abs{\int_{\tau}^{\tau+\lambda}N^2\mathcal{L}_N\langle\pi_s^N,G\rangle ds}\right]
	+\frac{{4}}{\epsilon^2}\mathbb{E}_{\mu_N}
	\left[\left(M_{\tau}^{N}(G)-M_{\tau+\lambda}^{N}(G)\right)^2\right].
	\end{split}\end{align*}
	Observe that since $G\in C^2([0,1])$ we have that {$ \abs{\Delta_N G(\tfrac{x}{N})}\leq 2\norm{G''}_{\infty} $ }
	and 
	$ \abs{\nabla_N^\pm G(\tfrac{x}{N})}\leq \norm{G'}_\infty $.
	{In particular, since}  there is at most one particle per site, {and recovering from \eqref{Dynkin'sFormula} the expression for $\mathcal{L}_n\langle\pi_s^N,G\rangle$, we obtain straightforwardly} 

	\begin{align*}\begin{split}
	|N^2\mathcal{L}_n\langle\pi_s^N,G\rangle|
	\lesssim 
	\tfrac{1}{N^{\theta-1}}\norm{G'}_{\infty}+\tfrac{1}{N^{\theta-1}}\norm{G''}_{\infty},
	\end{split}\end{align*} 
	where the notation $ \lesssim $ means "less than a constant times". 
	{As a consequence}, for $ \theta\geq1 $
	\begin{align*}\begin{split}
	\lim_{\gamma\to\infty}\limsup_{N\to\infty}\sup_{\tau\in\mathcal{T}_T,\lambda\leq\gamma} \mathbb{E}_{\mu^N}\left[\abs{\int_{\tau}^{\tau+\lambda}N^2\mathcal{L}_N\langle\pi_s^N,G\rangle ds}\right]=0.
	\end{split}\end{align*}
	
	Now we treat the remaining term. From Dynkin's formula, $ (M_t^{N}(G))^2-\langle M^{N}(G)\rangle_t $ is a (mean zero) martingale with respect to the natural filtration $ \mathcal{F}_t $.  From \cite{KL} (Appendix $ 1.6 $) one obtains that its quadratic variation is $ \langle M^{N}(G)\rangle_t:= \int_0^tB_s^{N}(G)ds $ , where 
	\begin{align*}\begin{split}
	B_s^{N}(G):=N^2\left(\mathcal{L}_N\inner{\pi^N(\eta_s),G}^2-2\inner{\pi^N(\eta_s),G}\mathcal{L}_N\inner{\pi^N(\eta_s),G}\right).
	\end{split}\end{align*}	
 {This yields}
		\begin{align*}\begin{split}
	\mathbb{E}_{\mu_N}\left[\left(M_{\tau}^{N,H}-M_{\tau+\lambda}^{N,G}\right)^2\right]=
	\mathbb{E}_{\mu_N}\left[\int_{\tau}^{\tau+\lambda}B_s^{N}(G)ds\right].
	\end{split}\end{align*}	

	{We can split  $ B_s^{N}(G):=B_{s,-}^{N}(G)+B_{s,0}^{N}(G)+B_{s,+}^{N}(G) $, where each term corresponds to the contribution of  $ \mathcal{L}_{N,-},\mathcal{L}_{N,0},\mathcal{L}_{N,+} $, respectively}.
	Now note that 
	\begin{align*}\begin{split}
	B_{s,0}^{N}(G)
	&=N^2\sum_{x\in\Lambda_N}\left(\inner{\pi^N(\eta_s^{x,x+1}),H}-\inner{\pi^N(\eta_s),H}\right)^2\\
	&=\sum_{x\in\Lambda_N}(\eta_s(x)-\eta_s(x+1))^2(G(\tfrac{x+1}{N})-G(\tfrac{x}{N}))^2 \leq \tfrac{N-1}{N^2} \norm{(G')^2}_{\infty} .
	\end{split}\end{align*}
	For the boundary dynamics, we bound the rates in the generator by a constant, {which yields} 
	\begin{align*}\begin{split}
	B_{s,-}^{N}(G)&\lesssim \frac{1}{N^{\theta-1}}\norm{G}_{\infty}^2\textrm {and }B_{s,+}^{N}(G)\lesssim\tfrac{1}{N^{\theta-1}}\norm{G}_{\infty}^2
	\end{split}\end{align*} {and concludes the proof.}
\end{proof}

\subsection{Characterization of the limit point}\label{sec:characterization}
{We now characterize the limit points  of $ \{\mathbb{Q}_N\}_{N\geq1} $ and show that they concentrate on trajectories satisfying the weak form of the hydrodynamic equation}. 
\begin{prop}\label{prop:charac}
For any limit point $ \mathbb{Q} $ of $ \{\mathbb{Q}_N \}_{N\geq1} $, it holds
	\[\mathbb{Q}\left(\pi_\cdot\in \mathcal{D}([0,T],{\mathcal{M}^+}):{R_T\mbox{ and }}F(\rho, G,t)=0	\right)=1\]
	where  {$R_T$ is the event on which $\pi_\cdot$ is absolutely continuous w.r.t the Lebesgue measure and with density in $\mathcal{H}^1$  {on the time segment $[0,T]$} and } $ F$ is given in \eqref{eq:Robin_integral} for $ \theta=1 $ and \eqref{eq:Neumann_integral} for $ \theta>1 $.
\end{prop}
\begin{proof}
We present the proof for the case $ \theta=1 $ and $K=2$, since for $ \theta>1 $ it is analogous. We present a remark at the end regarding the extension to other values of $K$. Fix a limit point $ \mathbb{Q} $ of $ \{\mathbb{Q}_N \}_{N\geq1} $. As a consequence of Corollary \ref{cor:EnergyEstimate},  we have that $\mathbb Q(R_T)=1$. To present the argument as simply and concisely as possible, assume that $G$ is time independent, but the same arguments apply when this is not the case. To prove the Proposition, we show that for any $  \delta>0 $ {and any $G\in C^2([0,1])$}:
	\begin{align}\label{limQ}
	&\mathbb{Q}\Big({R_T \mbox{ and }}\sup_{0\leq t\leq T}|F(\rho,G,t)|>\delta	
	\Big){=0}, 
	\end{align}  that is 
	\begin{align}\begin{split}\label{limQ_1}
	\mathbb{Q}&\left({R_T\mbox{ and }}
	\mid \sup_{0\leq t\leq T}\mid 
	\langle \rho_{t},  G\rangle  -\langle f_0 , G \rangle + \int_0^t\langle \rho_{s}, \triangle G\rangle \, ds  \right.+ \int^{t}_{0}  \Big \{\rho_{s}(1) \partial_u  G(1)-\rho_{s}(0)  \partial_u  G(0) \Big\} \, ds\\
	& -\int^{t}_{0}\Big\{G(1)(\beta_1-(\beta_1+\delta_1)\rho_{s}(1)+(\delta_2-\beta_2)(\rho_s^2(1)-\rho_s(1)))\,\Big\}ds \\ 
	&\left.- \int_0^t \Big\{ G(0)(\alpha_1-(\alpha_1+\gamma_1)\rho_s(0)+(\gamma_2-\alpha_2)(\rho_{t}^2(0)-\rho_t(0)) )\Big\}ds\mid>\delta	
	\right){=0.}
	\end{split}\end{align}
	Due to the boundary terms,  the  set inside the probability above is  not an open set in the Skorohod space. As a consequence, at this point, we cannot apply \textit{Portmanteau's theorem}.  To solve this problem, we take the following functions:
	$\overleftarrow{\iota}_\epsilon^u(v)=\frac{1}{\epsilon}1_{(u-\epsilon,u]}(v)
$ and $
	\overrightarrow{\iota}_\epsilon^u(v)=\frac{1}{\epsilon}1_{[u,u+\epsilon)}(v),
	$
	and we use the notation
	$\inner{\pi_s,\overleftarrow{\iota}_\epsilon^u}=\frac{1}{\epsilon}\int_{u-\epsilon}^u\rho_s(v)dv$ and $
	\inner{\pi_s,\overrightarrow{\iota}_\epsilon^u}=\frac{1}{\epsilon}\int_{u}^{u+\epsilon}\rho_s(v)dv.
	$
	Now observe that 
since $ \rho\in L^2(0,T;\mathcal{H}^1) $, it is easy to prove 
 for all $ \eps>0 $ that
	$\abs{\rho_s(u)-\inner{\pi_s,\overleftarrow{\iota}_\epsilon^u}}\leq \frac12\epsilon\norm{\partial_u \rho}_2^2$.
	As a consequence of the last result, we can bound the {probability on the left-hand side of \eqref{limQ_1} by} 
	\begin{align}\begin{split}\label{Q com iota}
	\mathbb{Q}&\Big({R_T\mbox{ and }}
	\sup_{0\leq t\leq T}\Big|
	\langle \rho_{t},  G_{}\rangle  -\langle f_0 , G_{} \rangle - \int_0^t\langle \rho_{s}, \partial^2_u G\rangle \, ds + \int^{t}_{0}  \Big \{
	\inner{\pi_s,\overleftarrow{\iota}_{\epsilon}^{1}}\partial_u  G(1)
	-\inner{\pi_s,\overrightarrow{\iota}_{\epsilon}^{0}}  \partial_u  G(0) \} \, ds\\
	&\hspace{1cm} -\int^{t}_{0}\Big\{G(1)(\beta_1-(\beta_1+\delta_1)\inner{\pi_s,\overleftarrow{\iota}_{\epsilon}^{1}}+(\delta_2-\beta_2)\inner{\pi_s,\overleftarrow{\iota}_{\epsilon}^{1}}
	(\inner{\pi_s,\overleftarrow{\iota}_{\epsilon}^{1-\epsilon}}
	-1)\,\Big\}ds\\
	& {-\int^{t}_{0}\Big\{G(0)(\alpha_1-(\alpha_1+\gamma_1)\inner{\pi_s,\overrightarrow{\iota}_{\epsilon}^{0}}+(\gamma_2-\alpha_2)\inner{\pi_s,\overrightarrow{\iota}_{\epsilon}^{0}}
	(\inner{\pi_s,\overrightarrow{\iota}_{\epsilon}^{\epsilon}}
	-1)\,\Big\}ds	\Big|>{\delta/2}\Big)+o_\varepsilon(1)}.
	\end{split}\end{align}

	To finally apply Portmanteau's theorem, we argue that we can approximate $ \overleftarrow{\iota}_\epsilon^u , \overrightarrow{\iota}_\epsilon^u $ by continuous functions in such a way that the error vanishes as $ \epsilon\to0 $. Then,  we apply Portmanteau's theorem 
	and {bound  the first term in \eqref{Q com iota} from above by $\liminf_{N\to\infty} \mathbb{Q}_N(A(T,G, \delta, \epsilon))$, where we shortened $A(T,G, \delta, \epsilon)$ for the event in \eqref{Q com iota}}.
Summing and subtracting $ \int_0^t \mathcal{L}_N\inner{\pi_s^N,{G}}ds $ {inside the absolute value in $A(T,G, \delta, \epsilon)$}, recalling  \eqref{Dynkin'sFormula}, {we obtain that $\mathbb{Q}_N(A(T,G, \delta, \epsilon))$ is less than the sum of the following contributions}
{
	\begin{equation*}
	P_1:=\mathbb{P}_{\mu^N}\left(
	\sup_{0\leq t\leq T} \abs{M^N_t}\geq \frac{\delta}{14}\right),
	\end{equation*}
	\begin{equation*}
	P_2:=\mathbb{P}_{\mu^N}\left(\sup_{0\leq t\leq T}
	\abs{\int_0^t\inner{\pi_s^N,\triangle_N G}- \langle \pi_{s}^N, \partial^2_u G\rangle ds}>
	\frac{\delta}{14}\right){,}
	\end{equation*}
	\begin{equation*}
	P_3:=\mathbb{P}_{\mu^N}\left(\sup_{0\leq t\leq T}
	\abs{\int_0^t\eta_{{sN^2}}(N-1)\nabla_N^-{G}(1)-
		\inner{\pi_s^N,\overleftarrow{\iota}_{\epsilon}^{1}}\partial_u  G(1)ds}>
	\frac{\delta}{14}\right),
	\end{equation*}
	\begin{equation*}
	P_4:=\mathbb{P}_{\mu^N}\left(\sup_{0\leq t\leq T}
	\abs{\int_0^tG(\tfrac{N-1}{N})(\beta_1-(\beta_1+\delta_1)\eta_{{sN^2}}(N-1))-
		G(1)(\beta_1-(\beta_1+\delta_1)\inner{\pi_s^N,\overleftarrow{\iota}_{\epsilon}^{1}})ds}
	>
	\frac{\delta}{14}\right),
	\end{equation*}
	\begin{multline*}
	P_5:=\mathbb{P}_{\mu^N}\left(\sup_{0\leq t\leq T}
	\Bigg|\int_0^t G(\tfrac {N-2}{N})(\delta_2\eta_{{sN^2}}(N-2)-\beta_2 \eta_{{sN^2}}(N-1)-	G(1)(\delta_2-\beta_2)\inner{\pi_s^N,\overleftarrow{\iota}_{\epsilon}^{1}}ds\Bigg|>\frac{\delta}{14}\right),
	\end{multline*}
	
	\begin{multline*}
	P_6:=\mathbb{P}_{\mu^N}\left(\sup_{0\leq t\leq T}
	\left|\int_0^t G(\tfrac{N-2}{N})(\delta_2-\beta_2)\eta_{{sN^2}}(N-1)\eta_{{sN^2}}(N-2)-\right.\right.\\
\left.	G(1)(\delta_2-\beta_2)\inner{\pi_s^N,\overleftarrow{\iota}_{\epsilon}^{1}}
	(\inner{\pi_s^N,\overleftarrow{\iota}_{\epsilon}^{1-\epsilon}}
	-1)ds\Bigg|
	>\frac{\delta}{14}\right),
	\end{multline*}
	and finally $P_7$, $P_8$, $P_9$ and $P_{10}$, which are the counterparts of $P_3$, $P_4$, $P_5$  and $P_6$ for the left boundary, which we do not explicitly write.}	
	 The first term $P_1$ can be estimated with Doob's inequality, 
\begin{align*}
	\mathbb{P}_{\mu^N}
	\left(
	\sup_{0\leq t\leq T}
	\abs{M_t^{N}}>\frac{\delta}{14}
	\right)&
	\leq \frac{C}{\delta}\mathbb{E}_{\mu^N}\left[\abs{M_T^{N}}^2\right]^{\frac{1}{2}}
	=\frac{C}{\delta} \mathbb{E}_{\mu^N}\left[\int_0^TB_s^{N,H}ds\right]^{\frac{1}{2}},	
\end{align*}
{where $C$ is a constant and  $B_s^{N,H}$ was introduced in Section \ref{sec:tight}. We can now proceed as in the proof of Proposition \ref{prop:tight} to show that $P_1$ vanishes, as $N\to\infty$.
	The second term $P_2$  vanishes, for any $N$ large enough, because $G$ is smooth.}

	To {estimate the remaining probabilities} we apply the replacement lemmas that are stated and proved in Appendix \ref{appendix:replacement}. To properly explain the procedure, recall from \eqref{eq:emp_aver} the definition of  $\overleftarrow{\eta}_{{sN^2}}^{\epsilon N}(N-1)$ and $	\overrightarrow{\eta}_{{sN^2}}^{\epsilon N}(1)$, since $ \overleftarrow{\eta}_{{sN^2}}^{\epsilon N}(N-1)=\inner{\pi_s^N,\overleftarrow{\iota}_\epsilon^1} $
	(resp. $ \overrightarrow{\eta}_{{sN^2}}^{\epsilon N}(1)=\inner{\pi_s^N,\overrightarrow{\iota}_\epsilon^0} $)
	we have $ \overleftarrow{\eta_{{sN^2}}}^{\epsilon N}(N-1)\sim\rho_s(1) $ (resp $ \overrightarrow{\eta_{{sN^2}}}^{\epsilon N}(1)\sim\rho_s(0) $), we show in Appendix \ref{appendix:replacement} that we can exchange  $ \eta_{{sN^2}}(N-1) $ (resp. $ \eta_{{sN^2}}(1) $) by the averages above, and $ \eta_{{sN^2}}(N-2) $ (resp. $ \eta_{{sN^2}}(2) $) by $ \eta_{{sN^2}}(N-1) $ (resp. $ \eta_{{sN^2}}(1) $). 
{In particular, to estimate $P_3$, note that the difference between $\nabla_N^-G(1)$ and $\partial_u G(1)$ is of order $N^{-1}$. Furthermore, applying Lemma \ref{replacement lemma_1} to $ \psi(\eta)=1 $, and using Markov's inequality, we can replace $ \eta_{{s}}(N-1) $ by $ \overleftarrow{\eta}_{{s}}^{\epsilon N}(N-1) $ at the cost of an error of order $ N^{-1} $. This proves that $P_3$ vanishes in the limit $N\to\infty$ and then $\epsilon \to 0$. $P_4$ and $P_5$ are estimated in the exact same way as  $P_3$.}

\medskip

{We now turn to $P_6$: in it, we first replace $ \eta_{{sN^2}}(N-2) $ by $ \overrightarrow{\eta}_{{sN^2}}^{\epsilon N}(N-1) $ by applying  Lemma \ref{replacement lemma_1} and Markov's inequality to $ \psi(\eta)=\eta_{{sN^2}}(N-1) $. Now that we have the term $ \eta_{{sN^2}}(N-1)\overrightarrow{\eta}^{\epsilon N} _{{sN^2}}(N-1)$, We apply Lemma \ref{replacement lemma_1} and Markov's inequality a second time to $ \psi(\eta)=\overrightarrow\eta_{{sN^2}}^{\epsilon N}(N-1) $ which allows us to replace $ \eta_{{sN^2}}(N-1) $ by $ \overrightarrow{\eta}_{{sN^2}}^{\epsilon N}(N-1) $ up to a vanishing error term.
Noting that $ \inner{\pi_s^N,\overrightarrow{\iota}_\epsilon^0}=\overrightarrow{\eta}_{{sN^2}}^{\epsilon N}(1)$ , and $  \inner{\pi_s^N,\overleftarrow{\iota}_\epsilon^1}=\overrightarrow{\eta}_{{sN^2}}^{\epsilon N}(N-1) $ and
		\begin{align*}
		\inner{\pi_s^N,\overrightarrow{\iota}_\epsilon^0}\inner{\pi_s^N,\overrightarrow{\iota}_\epsilon^\epsilon}=
		\overrightarrow{\eta}^{\epsilon N}_{{sN^2}}(1)
		\overleftarrow{\eta}^{\epsilon N}_{{sN^2}}(\epsilon N+1)+O((\epsilon N)^{-1}),
		\end{align*} proves as wanted that $P_6$ vanishes, in the limit $N\to\infty$ and then $\epsilon \to 0$.

The bounds for $P_7$, $P_8$, $P_9$ and $P_{10}$, are analogous to those on $P_3$, $P_4$, $P_5$  and $P_6$.
		Together, all those bounds prove that 
		$\limsup_{\epsilon\to 0}\limsup_{N\to\infty} \mathbb{Q}_N(A(T,G, \delta, \epsilon))=0$, so that \eqref{Q com iota} vanishes in the limit $\limsup_{N\to\infty}$ and then ${\epsilon\to 0}$. This proves Proposition \ref{prop:charac}}.
\end{proof}

	\begin{rem}[Case $ K\geq 2 $]
		For the general case $ K\geq 2 $, the main problem are the terms of the form $ \rho_s^{K-1}(0) $ and $ (1-\rho_s(0))^{K-1} $ (and similarly for the right boundary). A simple way to solve this is to proceed by induction. Since 
		$
		a^2=(a+ b_1-b_1)(a+ b_2-b_2)=(a-b_1)(a-b_2)+b_1(a-b_1)+b_2(a-b_2)+b_1b_2
		$
		and we have that $ b_1b_2a=b_1b_2(a+b_3-b_3)=b_1b_2(a-b_3)+b_1b_2b_3 $, taking $ a\equiv\rho_s(0) $ and $ b_j\equiv\inner{\pi_s,\overrightarrow{\iota}_\epsilon^{(j-1)\epsilon}} $ for $ j\geq0 $, we can replace $ \rho_s^{K-1}(0) $ by $ \prod_{j=0}^{K-2}\inner{\pi_s,\overrightarrow{\iota}_\epsilon^{j\epsilon}} $ plus a sum of terms that vanish when $ \epsilon\to0 $. For the right boundary the argument is analogous.
		\end{rem}
\section{Fick's Law}
\label{sec:currrent}

In this section, we prove Theorem  \ref{thm:current}. Recall the notations set in Section \ref{sec:currentsthm}. Fix $x\in \Lambda_N$. In order to apply Dynkin's formula to the current $J^N_t(x)$ and $K_t^N(x),$ we denote by  {
$\tilde {\mathcal L}^x$}
the joint generator of $\eta,$ $J^N_t(x)$ and $K_t^N(x),$ given by
\begin{equation}
\tilde{ \mathcal L}^x =\tilde{ \mathcal L}^x_{N,0}+\frac{1}{N^\theta}\tilde{ \mathcal L}^x_{N,b}
\end{equation}
where  
\begin{multline*}
\tilde {\mathcal L}^x_{N,0} f(\eta, J)=\sum_{z\in\Lambda_N\setminus \{x\}} (\eta(z)-\eta(z+1))(f(\eta^{z,z+1}, J)-f(\eta, J))\\
+\eta(x)(1-\eta(x+1))(f(\eta^{x,x+1}, J+1)-f(\eta, J))+\eta(x+1)(1-\eta(x))(f(\eta^{x,x+1}, J-1)-f(\eta, J)){,}
\end{multline*}
and $\tilde{ \mathcal L}^x_{N,b}=\tilde{ \mathcal L}^x_{N,+}+\tilde{ \mathcal L}^x_{N,-}$, with
  \begin{multline*}
\tilde {\mathcal L}^x_{N,\pm} f(\eta, K)=\sum_{z\in I_{\pm}^K\setminus \{x\}} c_z^{\pm}(\eta)(f(\eta^{z}, K)-f(\eta, K))\\
+\mathbb{1}_{\{x\in I^K_\pm\}}\Big[(1-\eta(x))c_x^{\pm}(\eta)(f(\eta^{x}, K+1)-f(\eta, K))+\eta(x)c_x^{\pm}(\eta)(f(\eta^{x}, K-1)-f(\eta, K))\Big].
\end{multline*}
Recalling the definition of the operators $ D^{N,\pm}_{\alpha,\gamma} $ and $ D^{N,\pm}_{\beta,\delta} $ in \eqref{def:D^N},  
\[\tilde {\mathcal L}_{N,0}^x J_s^N(x)=\eta_s(x)-\eta_s(x+1), \quad \tilde {\mathcal L}_{N,\pm}^x K_s^N(x)=(D^{N,-}_{\alpha,\gamma}\eta_s)(x) \quad \mbox{ and }\tilde {\mathcal L}_{N,b}^x K_s^N(x)=(D^{N,+}_{\beta,\delta}\eta_s)(x),\]
for $x\in \Lambda_N$, $I^K_-$, $I^K_+$, respectively.

By Dynkin's  formula, \begin{equation*}
\widehat M_t^N(x):=J^N_t(x)-J_0^N(x)-\int_0^{tN^2}(\eta_s(x)-\eta_s(x+1)) \, ds,
\end{equation*}
is a martingale w.r.t. $\mathcal F_t$, so that
for any  test function $f\in C^1(0,1)$ 
\begin{equation*}
{\widetilde{M}_t^N(f)}:=J^N_t(f)-J_0^N(f)-
\int_0^{tN^2}\sum_{x=1}^{N-2}f(\tfrac{x}{N})(\eta_s(x)-\eta_s(x+1))  \, ds,
\end{equation*}
is a martingale as well. By summation by parts,  the {time integral above rewrites} 
\begin{equation}\label{eq:curr_1}
\begin{split}
&\int_0^{{tN^2}}\frac{1}{N}\sum_{x=1}^{N-1}\nabla_N^-f(\tfrac{x}{N})\eta_s(x)+f(\tfrac{0}{N})\eta_s(1)-f(\tfrac{N-1}{N})\eta_s(N-1)\, ds\\=&\int_0^t\langle\pi_t^N, \nabla_N^-f\rangle\,ds+\int_0^{{tN^2}}f(\tfrac{0}{N})\eta_s(1)-f(\tfrac{N-1}{N})\eta_s(N-1)\, ds.
\end{split}
\end{equation}
A simple computation also based on Dynkin's formula, shows that its quadratic variation is given by
$\langle {\widetilde M^N(f)}\rangle_t=\int_0^t \frac{1}{N^2}\sum_{x\in\Lambda_N}f^2(\tfrac xN)(\eta_s(x)-\eta_s(x+1))^2ds,$
so that the martingale $\widetilde{M}_t^N(f)$ vanishes  in $\bb L^2(\mathbb P_{\mu_N})$, as ${N}\to\infty$. 

Now we analyze the {time integral} above. From  Theorem \ref{th:hyd_ssep} and the Replacement Lemma \ref{replacement lemma_1}
 the expression \eqref{eq:curr_1} converges,  as $N\to\infty$, in $\mathbb P_{\mu_N}$ to  \begin{equation}
\begin{split}
\int_0^t\int_0^1 \partial_u f(u)\,  \rho_s(u)\, du\,ds+\int_0^t f(0)\rho_s(0)-f(1)\rho_s(1)\, ds
=-\int_{0}^t\int_0^1 f(u)\partial_u \rho_s(u)\,ds.
\end{split}
\end{equation} 
Now we look at the non-conservative current. Using the same argument as above, we have that 
\begin{equation*}
\begin{split}
K^N_t(f)-K_0^N(f)-&
N^{1-\theta}\int_0^{{tN^2}}\sum_{x\in I_-^K}f(\tfrac xN)
(D_{\alpha,\gamma}^{N,-}\eta_{{s}})(x) ds
-N^{1-\theta} \int_0^{{tN^2}} \sum_{x\in I_+^K}f(\tfrac xN)
(D_{\beta,\delta}^{N,+}\eta_{{s}})(x) ds
\end{split}
\end{equation*}
is a martingale. As above it can be shown that this martingale vanishes  in $\bb L^2(\mathbb P_{\mu_N})$, as ${N}\to\infty$. 
 When $\theta>1$ it is easy to see that the integral term above vanishes as ${N}\to+\infty$. In the case $\theta=1 $,  from repeatedly applications of  the replacement lemmas stated in  Appendix \ref{sec:repl_lemma} the last expression converges, w.r.t. $\mathbb P_{\mu_N}$, as
 $N\to\infty$,   to 
 \begin{align*}
\int_0^t f(0)(D_{\alpha,\gamma}\rho_s)(0)+f(1)(D_{\beta,\delta}\rho_s)(1)\,ds,
\end{align*}
which finishes the proof.

\section{Hydrostatic limit}\label{sec:hs}

\subsection{ Proof of  Theorem \ref{thm:hydrostatics}}
\label{sec:proofhydrostatics}

Given an integer $N$, and $\theta\geq 1$, define the distribution $ \mathcal{P}_N$ of the stationary empirical measure, on $ \mathcal{M}^+ $, as
$\mathcal{P}_N:=\mu_N^{ss}\circ (\pi_N)^{-1}$, where both $ \mathcal{M}^+ $ and the empirical measure $\pi_N$ were introduced at the beginning of Section \ref{sec:HL}. Recall from Definition \ref{def:weak_sol_ Robin} the definition of the functional $F(\rho, G, t)$ (depending on $\theta$). Define 
\[ {\mathcal{E}_T}:=\left\{\pi\in\mathcal{M}^+:\pi(du)=\rho^*(u)du \textit{ $ , $ } F(\rho^*,G,t)=0, \;\forall t\in [0,T],\; \forall G \in C^{1,2} ([0,T]\times[0,1]) \right\}, \]
which is the set of weak stationary solutions to the hydrodynamic equation.
{Since $T$ is fixed, from here on, we will drop the subscript and simply write $\mathcal E=\mathcal E_T$.}
Now let $ d $ be the distance defined on { $ \mathcal{M}^+ $} under which this space is a Polish space (see \cite{KL}, Chapter $ 4 $ for an example). The following result, which is analogous to e.g. Theorem $ 2.2 $ in \cite{LT}, is the main ingredient to prove Theorem \ref{thm:hydrostatics}.
\begin{prop}\label{hsconcentration}
	$ \{\mathcal{P}_N\}_{N\in\mathbb{N}} $ is concentrated in $ \mathcal{E} $, \textit{i.e.,} $ \forall \delta>0 $,
	\begin{align*}
	\lim_{N\to\infty}
	\mathcal{P}_N
	\left(
	{\pi}\in\mathcal{M}^+:\inf_{\tilde{\pi}\in\mathcal{E}}d({\pi},\tilde{\pi})\geq\delta
	\right)=0.
	\end{align*}
\end{prop}
To prove Proposition \ref{hsconcentration}, one needs two ingredients: 
\begin{enumerate}[i)]
\item The empirical measure  {is} macroscopically governed by a hydrodynamic equation (i.e. the hydrodynamic limit, Theorem \ref{th:hyd_ssep}, proved in Section \ref{sec:hyd}). 
\item The existence of a unique solution of \eqref{eq:Robin_equation} (cf. Lemma \ref{lem:uniqueness}) and its convergence, w.r.t. the $\bb L^2$ norm, as time goes to infinity, to a stationary solution, which is a consequence of Proposition \ref{prop:statio_conv} below for $\theta=1$. In the case $\theta>1$ this convergence is classical, but  the interested reader can straightforwardly adapt the argument we present below for the Robin  case , derived when $\theta=1$, to the Neumann  case, derived when $\theta>1$.
\end{enumerate}
We will not prove this proposition, because once those two ingredients are obtained, it is a straightforward adaptation of  Theorem $ 2.2 $ in \cite{LT}.

We now use Proposition \ref{hsconcentration} to prove Theorem \ref{thm:hydrostatics}. In the case $\theta=1$, under assumptions  \eqref{H0} and \eqref{H2}, we check in Section \ref{sec:statio_uniq} that the set
$\mathcal{E} $ above is a singleton (more precisely, $\mathcal{E}=\{\rho^*(u)du\} $ where $ \rho^*(u)=(1-u)\rho^*(0)+u\rho^*(1) $ with its value at the boundary determined by the unique solution of the nonlinear system of equations $ \rho^*(1)-\rho^*(0)=-D_{\alpha,\gamma}\rho^*(0)={D_{\beta,\delta}}\rho^*(1) $), so that the first assertion of Theorem \ref{thm:hydrostatics} follows immediately from Proposition  \ref{hsconcentration}. 
For details we refer the reader to  \cite{Tsu}.

\medskip

We now turn to the second assertion of Theorem \ref{thm:hydrostatics}, i.e. the case $\theta>1$. In this case, {since the hydrodynamic equation is governed by the heat equation with Neumann boundary conditions,} any constant solution is a stationary solution to the hydrodynamic equation, so that
\[ \mathcal{E}:=\left\{\pi\in\mathcal{M}^+:\pi(du)=m du, \; m\in [0,1]\right\}. \]
{An intriguing question is to find the particular constant that the system choses to stabilize.} Therefore, we further need to prove that  $\mathcal{P}_N$ only charges, in the limit $N\to\infty$, a single value $m^*$ from all the possible constant values. For that purpose we use the method developed  in \cite{Tsu} {(therein applied to the present model with $ K=1 $)}, which consists 
in studying the evolution of the process started from its stationary state in a subdiffusive time scale $ N^{1+\theta}$.  In this subdiffusive time scale, the total mass of the system evolves via the boundary dynamics {and is macroscopically described by an ordinary differential equation. Moreover, the time-independent solution to this ODE} 
is exactly the constant  $m^*$ we are looking for. The non-linear boundary terms pose some extra technical difficulties w.r.t.  \cite{Tsu}, which are solved in Corollary \ref{replacement lemma_mass}.

{To simplify the exposition we consider the case $K=2$, and we make  the appropriate remarks in  the general case of $K$.

We now consider the process on the \emph{subdiffusive time-scale} $\{\eta^N_t:=\eta_{tN^{1+\theta}}\}_{t\geq0} $, with initial distribution $ \mu_N^{ss}$.
For each $ t\geq0 $ and $ \theta>1 $ fixed, we define the mass of the system  as
	\begin{align}
	\label{eq:defmt}
	m^N_t=\dfrac{1}{N-1}\sum_{x\in\Lambda_N}{\eta^N_{t}}(x),
	\end{align}
	and for each $ T>0 $ we let $ \mathcal{D}([0,T],\mathbb{R}) $ be the set of c\'adl\'ag trajectories $ m_\cdot:[0,T]\to\mathbb{R} $ with respect to the Skorohod topology. For each $ N\in\mathbb{N} $, denote by $ \mathcal{Q}_N $ the distribution of the trajectory $ (m^N_t)_{t\in [0,T]} $ on $ \mathcal{D}([0,T],\mathbb{R}) $, with $\eta$ started from the stationary distribution $\mu^{ss}_N$.  
A straightforward adaptation (to the subdiffusive timescale) of Proposition \ref{prop:tight} shows that the sequence $ \{\mathcal{Q}_N\}_{N\geq1} $ is tight.  From the stationarity of $ \mu_N^{ss} $, for any $t\geq0$, and any $I=(a,b)\subset[0,1]$,
\[ \mathcal{P}_N(\pi^N_\cdot:\inner{\pi^N,1}\in I)=Q_N(m_\cdot^N:m_t\in I).\]
One may then bound, for any limit points $\mathcal{P}^*$, $\mathcal{Q}^*$ of the sequences  $\{\mathcal{P}_N\}_{N\geq1}$, $ \{\mathcal{Q}_N\}_{N\geq1}$ (for more details, cf. \cite{Tsu}, p.11)  and for any fixed time $t\geq 0$,
\begin{equation}
\label{eq:PQstar}
\mathcal{P}^*(\pi(du)=mdu, \;m\in I)\leq \mathcal{Q}^*(m_\cdot:m_t\in \overline{I}).
\end{equation}
To conclude, we now only need to prove the following result
\begin{lem}
\label{lem:concricatti}
 There exists $m^*\in [0,1]$ such that , for any $\varepsilon>0$ 
\begin{equation}
\label{eq:convRicatti}
\mathcal{Q}^*(m_\cdot:|m_t-m^*|\geq \varepsilon)\underset{t\to\infty}{\longrightarrow}0{.}
\end{equation}
\end{lem}

\begin{proof}[Proof of Lemma \ref{lem:concricatti}]{
We start by characterizing the typical trajectories of $\mathcal{Q}^*$.
We apply Dynkin's formula as in \eqref{Dynkin'sFormula} and we take $G\equiv 1$ to obtain that }
\begin{multline}
\label{eq:dynkin_mass}
m^N_t=m^N_0+M_t^N(1)+
\int_0^{{t}}
\alpha_1+\beta_1
+\eta_s^N(1)(\alpha_2-(\alpha_1+\gamma_1))
+\eta_s^N(N-1)(\beta_2-(\beta_1+\delta_1))ds
\\
-\int_0^{{t}}\gamma_2\eta_s^N(2)+\delta_2\eta_s^N(N-2)
+\eta_s^N(1)\eta_s^N(2)(\alpha_2-\gamma_2)
+\eta_s^N(N-1)\eta_s^N(N-2)(\beta_2-\delta_2)ds.
\end{multline}
Simple computations similar to the ones used in the proof of Proposition \ref{prop:tight} also show that {the quadratic variation of  $M_t^N (1)$ is of order $ O(N^{-1}) $} so that  $ M_t^N (1)$ vanishes as $ N\to\infty $, with respect to the $\mathbb L^2(\mathbb P_{\mu^{ss}_N})$-norm. Moreover, from Corollary \ref{replacement lemma_mass} we are able to replace $ \eta_{s}^N(1)$ and $\eta_{s}^N(2) $ (resp. $ \eta_{s}^N(N-1) $ and $ \eta_{s}^N(N-2) $) by $ m^N_s $. From this we get that
\begin{multline}\label{eq:dynkin_mass_new}
m^N_t=m^N_0+M_t^N(1)+
\int_0^{tN^{1+\theta}}
\alpha_1+\beta_1
+m^N_s(\alpha_2-(\alpha_1+\gamma_1))
+m^N_s(\beta_2-(\beta_1+\delta_1))ds
\\-\int_0^{{tN^{1+\theta}}}\gamma_2m^N_s+\delta_2m^N_s
+(m^N_s)^2(\alpha_2-\gamma_2)
+(m^N_s)^2(\beta_2-\delta_2)ds.
\end{multline}
To simplify notation let $ \mathfrak i_x=\alpha_x+\beta_x $ and $ \mathfrak o_x=\gamma_x+\delta_x $, for $ x=\{1,\dots,K\} $. By taking the limit, when $N\to+\infty$, in the previous identity, we obtain that any limit point $\mathcal{Q}^*$ is concentrated on solutions of the Ricatti Equation
\begin{align}\label{diffeq}
\mathcal{Q}^*\pa{m_\cdot:  \;m_t=m_0+ \int_0^t
\mathfrak{i}_{1}+(\mathfrak{i}_{2}-\mathfrak{o}_{2}-(\mathfrak{i}_{1}+\mathfrak{o}_{1}))m_s-(\mathfrak{i}_{2}-\mathfrak{o}_{2})m^2_sds}=1.
\end{align}

\begin{rem}
Observe that the equation above is equivalent to $ m_t=m_0+\int_0^tD_{\mathfrak i,\mathfrak o}m_sds $. In fact, as for the hydrodynamic limit, the same arguments show that for general values of $ K>2 $ we have an analogous integral equation, where now $ D_{\mathfrak i,\mathfrak o} $ induces a $ K-$degree polynomial. For the case  $ K>2 $, the proof above is indeed identical in this case as long as we assume  \eqref{H3}. 
 The only extra technical difficulty is that in  \eqref{eq:dynkin_mass} we shall have higher order polynomials in $\eta$, so that we have to apply Corollary \ref{replacement lemma_mass} a certain number of times to get closed equations in $m^N_s$. 
\end{rem}
To conclude the proof of Lemma \ref{lem:concricatti}, we need to show the uniqueness of solutions for the Ricatti equation and that all solutions  uniformly converge to the same constant.  We first state the following technical lemma. 
\begin{lem}\label{lem:D1-D2}
	Let $ \lambda=(\lambda_1,\dots,\lambda_K) $ and $ \sigma=(\sigma_1,\dots,\sigma_K) $ and $ D_{\lambda,\sigma} $ be the operator defined in \eqref{exp:D} for $ K\geq1 $ fixed. Then for $ f_i:\mathbb{R}\to{[0,1]} $ with $ i=1,2 $, we have
	\begin{align}\label{D_decomp}
		D_{\lambda,\sigma}f_1-D_{\lambda,\sigma}f_2=-(f_1-f_2)V_{\lambda,\sigma}(f_1,f_2),
	\end{align}
	where $ V_{\lambda,\sigma}(f_1,f_2)=V_{\lambda}(f_1,f_2)+V_{\sigma}(1-f_1,1-f_2) $ with the operator $ V_\phi $, for any sequence $ \phi=(\phi_x)_{1\leq x\leq K} $, is acting on functions { $(f_1,f_2):\mathbb{R}\times \mathbb{R}\to\mathbb{R}^+_0\times\mathbb{R}^+_0 $, as
	\begin{align*}
		V_\phi (f_1, f_2)(u_1, u_2)=\sum_{x=1}^K(\phi_x-\phi_{x+1})v_{x}(f_1(u_1),f_2(u_2))
	\end{align*}
	where {$\phi_{K+1}:=0$ by convention and} {
	\begin{align*}
		v_x(y,z)=\begin{cases}
		1,&x=1,\\
		y+z,&x=2,\\
		y^{x-1}+z^{x-1}+
		\sum_{i=1}^{x-2}z^iy^{x-1-i},& 3\leq x\leq K.
		\end{cases}
	\end{align*}
}}
	In particular, if  $ \lambda $ and $ \sigma $ are {positive}, non-increasing and {$ f_i\in(0,1) $ for $ i=1,2 $ } then there is a constant $ \underline{v}_K(\lambda,\sigma)>0 $ such that 
	\begin{align}\label{v>0}
		V_{\lambda,\sigma}f\geq \underline{v}_K(\lambda,\sigma).
	\end{align}
{	Moreover, if $ f_i\in(0,1) $, $ \lambda $ (resp. $ \sigma $) is the constant zero sequence and $ \sigma $ (resp. $ \lambda $) is non-negative and non-increasing, $ V_{\lambda,\sigma}f $ is also bounded from below by a positive constant.}
\end{lem}

We can now conclude the proof of Lemma \ref{lem:concricatti}. From Lemma \ref{lem:D1-D2} it is simple to see that $ m $ is locally Lipschitz. By iteration we can extend to all times up to time $ t>0 $, thus showing uniqueness. 
	To see that there exists  a unique solution $ m\equiv m^*\in[0,1] $ to the equation defined by the mass, 
	$ D_{\mathfrak {i},\mathfrak{o}} m $, observe that an immediate consequence of Lemma \ref{lem:D1-D2} is that 
	$ D_{\mathfrak {i}, \mathfrak {o}}f $ is both continuous and monotone decreasing on $ f\in[0,1] $. In particular, letting  $\textbf{0}$ (resp. $ \textbf{1}$ ) be the constant $ 0 $ (resp. constant $ 1$) function on $ [0,1] $, \textcolor{red}{and recalling  \eqref{H1}, ie $ \mathfrak{i}_1\neq 0$ and $\mathfrak{o}_1\neq0 $,} we have that 
	\begin{align*}
	-\mathfrak {o}_1=D_{\mathfrak {i}, \mathfrak {o}}\textbf {1}\leq D_{\mathfrak {i}, \mathfrak {o}}f\leq D_{\mathfrak {i}, \mathfrak {o}}\textbf {0}=\mathfrak {i}_1,
	\end{align*}
	thus, by the intermediate value theorem there exists a (unique)	{$ f\equiv m^*\in[0,1] $ such that  $ D_{\mathfrak {i}, \mathfrak{o} }m^*=0 $}.

Now fix a solution $m $ to  $$m_t=m_0+\int_0^tD_{\mathfrak i,\mathfrak o}m_sds, $$ and define $\widehat{m}=m-m^*$, \eqref{D_decomp}  yields 
\[\widehat{m}_t=\widehat{m}_0-\int_0^t\widehat{m}_sV_{\mathfrak i,\mathfrak o}(m_s,m^*)ds,\] 
so that $ \widehat{m}_t=\widehat{m}_0 e^{-\int_0^tV_{\mathfrak i,\mathfrak o}(m_s,m^*)ds}.$ Using  \eqref{v>0} then yields  $|\widehat{m}_t|\leq e^{-\underline{v}_K(\mathfrak i,\mathfrak o)t}$ which proves \eqref{eq:convRicatti}.
\end{proof}
We now prove the technical Lemma.
\begin{proof}[Proof of Lemma \ref{lem:D1-D2}]
	For $ u_1,u_2\in\mathbb{R} $, let $ y:=f_1(u_1) $ and $ z:=f_2(u_2) $. Then
 \begin{align*}
 D_{\lambda,\sigma}y-D_{\lambda,\sigma}z
 =\sum_{x=1}^K\lambda_x
 \Big(
 (1-y)y^{x-1}-(1-z)z^{x-1}
 \Big)
 -\sigma_x
 \Big(
 y(1-y)^{x-1}-z(1-z)^{x-1}
 \Big).
 \end{align*}
 Observing that 
 \begin{align*}
 (1-y)y^{x-1}-(1-z)z^{x-1}&=(y^{x-1}-z^{x-1})-(y^x-z^x)
 ,\\
 y(1-y)^{x-1}-z(1-z)^{x-1}&=
 -\Big\{
 ((1-y)^{x}-(1-z)^{x})-((1-y)^{x-1}-(1-z)^{x-1})
 \Big\}.
 \end{align*}
{For $y,z$ fixed,  let } $ h_0(x)=y^x-z^x $ and $ h_1(x)=(1-y)^x-(1-z)^x $. Performing a summation by parts, we have
 \begin{align}\label{eq:pre_D}
\begin{split}
 D_{\lambda,\sigma}y-D_{\lambda,\sigma}z
 &=
 \sum_{x=0}^{K-1}\sigma_{x+1}(h_1(x+1)-h_1(x))
 -\sum_{x=0}^{K-1}
 \lambda_{x+1}(h_0(x+1)-h_0(x))
 \\&=-(\lambda_Kh_0(K)-\sigma_Kh_1(K))
 -\sum_{x=1}^{K-1}
 \Big\{
 (\lambda_x-\lambda_{x+1})h_0(x)-(\sigma_x-\sigma_{x+1})h_1(x)\Big\}.
\end{split}
 \end{align} 
{Since for  any integer $ x\geq3 $ it } { holds
	 \begin{align*}
	a^x-b^x=(a-b)\left(a^{x-1}+b^{x-1}+
	\sum_{i=1}^{x-2}b^ia^{x-1-i}
	\right),\qquad a\geq b\geq 0,
	\end{align*} 
 we have the following decomposition
 \begin{align*}
 h_0(x)=(y-z)
 \left\{
 1_{x=1}+(y+z)1_{x=2}+1_{x\geq 3}
 \left(y^{x-1}+z^{x-1}+
 \sum_{i=1}^{x-2}z^iy^{x-1-i}
 \right)
 \right\}.
 \end{align*}
 Note that for $ K=2 $ we do not have the sum over $ i $ above, and if either $ y $ or $ z $ equal $ 0 $, the identity trivially holds. For {$ h_1, $} we  replace $ y $ and $ z $ by $ 1-y $ and $ 1-z $, respectively. For $ y,z\in (0,1) $, replacing $ h_0(x) $ and $ h_1(x) $ in \eqref{eq:pre_D} and rearranging terms, the proof of \eqref{D_decomp} ends. For  \eqref{v>0}, since 
 \begin{align*}
 	V_{\lambda,\sigma}(y,z)=\lambda_1-\lambda_2+\sigma_1-\sigma_2
 	+\sum_{\substack{x=2\\ \sigma_{K+1},\lambda_{K+1}:=0}}^K(\lambda_x-\lambda_{x+1})v_{x}(y,z)+(\sigma_x-\sigma_{x+1})v_{x}(1-y,1-z),
 \end{align*}
 with $ v_x $ defined as in the statement, for decreasing sequences $ \lambda,\sigma $ the sum is bounded  from below by zero, while $ \lambda_1-\lambda_2+\sigma_1-\sigma_2>0 $. For constant sequences, \textit{i.e.,} $ \lambda_x=\lambda_1,\sigma_x=\sigma_1 $ for $ x=2,\dots,K $, we have only the last term of the sum, $ V_{\lambda,\sigma}(y,z)=\lambda_1v_K(y,z)+\sigma_1v_K(1-y,1-z) $, which is bounded from below by a positive constant if $ y,z\neq 0,1 $ and either $ \lambda_1 $ or $ \sigma_1 $ is not zero.
}	
%
%

\end{proof}
We end the present section with some observations. 
{\begin{rem}
	Note that $ v_x(y,z) $, for $ y,z\in\{0,1\} $ is well defined for any $ x=1,\dots,K $. Moreover, we have $ v_x(0,0)=1_{x=1},\; v_x(1,0)=v_x(0,1)=1 $ and $ v_x(1,1)=x $. This implies that we can indeed find positive lower bounds for $ V_{\lambda,\sigma}(y,z) $ with $ y,z\in\{0,1\} $ under some minor restrictions. More precisely, $ V_{\lambda,\sigma}(1,0)=V_{\lambda,\sigma}(0,1)=\lambda_K+\sigma_K $ (enough to consider $ \lambda_K>0 $ or $ \sigma_K>0 $), while $ V_{\lambda,\sigma}(0,0)=\lambda_1-\lambda_2+\sum_{x=1}^Kx(\sigma_x-\sigma_{x+1}) $ (enough to consider $ \lambda,\sigma $ constant (not zero) or decreasing; $ \sigma $ constant sequence $ 0 $ and $ \lambda_1\neq \lambda_2  $; or $ \lambda_1=\lambda_2 $ and $ \sigma $ constant (not zero) or decreasing ), while for the restrictions where $ V_{\lambda,\sigma}(1,1)>0 $ it is enough to consider the restrictions for $ V_{\lambda,\sigma}(0,0) $ with $ \lambda $ and $ \sigma $  exchanged.
\end{rem}
}
\begin{rem}[Case $ K=2 $.]
\ccl{Fixing $ K=2 $, the Ricatti equation with constant coefficients in  \eqref{diffeq}, that is
	\begin{align*}
		m_t=m_0+ \int_0^t
		\mathfrak{i}_{1}+(\mathfrak{i}_{2}-\mathfrak{o}_{2}-(\mathfrak{i}_{1}+\mathfrak{o}_{1}))m_s-(\mathfrak{i}_{2}-\mathfrak{o}_{2})m^2_sds,
		\end{align*}
		is explicitly solvable. For simplicity, let us write 
	\begin{equation*}
	a=-(\mathfrak{i}_{2}-\mathfrak{o}_{2}),\qquad
	b=\mathfrak{i}_{2}-\mathfrak{o}_{2}-(\mathfrak{i}_{1}+\mathfrak{o}_{1}),\qquad
	c=\mathfrak{i}_{1},
	\end{equation*}
so that  $ \frac{dm_t}{dt}=am_t^2+bm_t+c $. If $a=0$, one obtains the explicit expression $m_t=(m_0+\frac{c}{b})e^{bt}-\frac{c}{b}$, which converges to $m^*:=-c/b=\mathfrak{i}_{1}(\mathfrak{i}_{1}+\mathfrak{o}_{1})$ because   $ b\leq 0$ by assumption \eqref{H1}.

We now assume $a\neq 0$. Let 	$\Delta=b^2-4ac$. One easily checks that $\Delta=(\mathfrak{i}_{1}-\mathfrak{o}_{1}+\mathfrak{i}_{2}-\mathfrak{o}_{2})^2+4\mathfrak{i}_{1}\mathfrak{o}_{1}>0$ thanks to assumption \eqref{H1} and that $ 2a+b,b\leq 0$ by assumption \eqref{H3} which guarantees  $\mathfrak{i}_{1}\geq \mathfrak{i}_{2}$, $\mathfrak{o}_{1}\geq \mathfrak{o}_{2}$.
	Define $ \kappa^\pm=\tfrac{ -b\pm \sqrt{\Delta}}{2a} $, we find that $ am_s^2+bm_s+c=a(m_s-\kappa^+)(m_s-\kappa^-)$.
	Observe that 
	\begin{equation}\label{eq:decompms}
	\frac{1}{am_s^2+bm_s+c}
	=\frac{1}{a}\frac{1}{(m_s-\kappa^-)(m_s-\kappa^+)}
	=\frac{1}{\sqrt{\Delta}}\left\{\frac{1}{m_s-\kappa^+}
	-\frac{1}{m_s-\kappa^-}\right\}.
	\end{equation}
	Assume that $m_0\neq{\kappa^\pm}.$
	Since $m_t-\kappa^\pm$ cannot change sign in time, (otherwise it would cross the value $\kappa^\pm$ which is a fixed point for the Ricatti equation)
	\begin{align*}
	\int_0^t\frac{dm_s}{m_s-\kappa^\pm}
	=\log\frac{m_t-\kappa^\pm}{m_0-\kappa^\pm}.
	\end{align*}
After some computations, this, together with \eqref{eq:decompms}, finally yields
	\begin{equation*}
		m_t=
		\frac{\kappa^--\varepsilon_t \kappa^+}{1-\varepsilon_t}=\kappa^++\frac{\kappa^--\kappa^+}{1-\varepsilon_t}, \quad  \quad \mbox{where} \quad \varepsilon_t=\frac{\kappa^--m_0}{\kappa^+-m_0}e^{-\sqrt{\Delta}t}.
	\end{equation*}
In particular, $m_t$ is monotonous, starts at $m_0$, and converges to $m^*:=\kappa^-$. We now only need to check that $\kappa^-\in [0,1]$, which is straightforward by separating the cases $a> 0$ and $a< 0$, and using that $b, \; 2a+b\leq 0$.}
\end{rem}
\begin{rem}[Case $ K> 2 $] 
To finish, we observe that for the model introduced in \cite{DMP12}, \textit{i.e.,} taking above $ \beta_x=\gamma_x=j $ and $ \alpha_x=\delta_x=0 $ for $ x=\{1,\dots,K\} $, a simple computation shows that the solution of $ D_{\mathfrak i,\mathfrak i}m^*=0 $ is, in fact, $ m^*=1/2. $ The model in \cite{DMP12} is a particular case of considering $ \mathfrak i=\mathfrak o $. For the latter, $ D_{\mathfrak i,\mathfrak i}m^*=0 $ can be solved for $ m\equiv m^* $ by noticing that $ D_{\mathfrak i,\mathfrak i}m=D_{\mathfrak i,0}m-D_{\mathfrak i,0}(1-m)=-(m-(1-m))V_{\mathfrak i,0}(m,1-m)=0\Leftrightarrow m^*=1/2 $. 
Perhaps more interesting is the case when $ \mathfrak i_x=\mathfrak i $ , $ \mathfrak o_x = \mathfrak o $ for $ x\in\{1,\dots, K\} $ , \textit{i.e.,} the rates are constant in $x$. Under these conditions, we have 
	\begin{align*}
		\frac{1-(m^*)^K}{1-(1-m^*)^K}=\frac{\mathfrak o}{\mathfrak i}.
	\end{align*}
If $ \mathfrak o = \mathfrak i $ then the mass stabilizes to the middle of point of $ [0,1] $, but for fixed $ \mathfrak i $ (resp. $ \mathfrak o $), as the rate of removal (resp. injection) of particles increases (resp. decreases)  that is $ \mathfrak o \nearrow $ (resp. $ \mathfrak i\searrow $), then the mass decreases \textit{exponentially} in $ K $.
\end{rem}

\subsection{On the uniqueness of the stationary macroscopic profile for $ \theta=1 $}
\label{sec:statio_uniq}

We now prove that, under our assumptions, the hydrodynamic limit for $ \theta=1 $ admits a unique stationary solution.

The same idea used for $ \theta>1 $ can be used now to guarantee uniqueness for the stationary solution  of the  heat equation with Robin boundary conditions. {Throught the proof we will assume \eqref{H2}, but we divide in two cases. First assume that $ \alpha_1\wedge\delta_1=\delta_1 $ and $ \gamma_1\wedge\beta_1=\beta_1 $. Start} by observing that for every $ \rho^*(1) $ fixed, there exists  a unique solution for $ \rho^*(0) $ of the equation $ -D_{\alpha,\gamma}\rho^*(0)=D_{\beta, \delta}\rho^*(1) $, which is equivalent to $ D_{\alpha,\gamma}\rho^*(0)=D_{\delta,\beta}(1-\rho^*)(1) $. From the previous observations we know that $ \forall f\in[0,1] $ we have $ D_{\delta,\beta}f\in[-\beta_1,\delta_1] $, and also that $ \forall u\in[-\gamma_1,\alpha_1] $ there is one $ f\in[0,1]: D_{\alpha,\gamma}f=u. $ In this way, there exists a unique $ \rho^*(0)\in[0,1] $ such that, for any fixed $ \rho^*(1) $ with $ D_{\delta,\beta}(1-\rho^*)(1)=u\in[-(\gamma_1\wedge\beta_1),\alpha_1\wedge\delta_1], $ we have $ D_{\alpha,\gamma}\rho^*(0)=u $. More precisely, by monotonicity, we have that 
	\begin{align*}
		\rho^*(0)&\in[D^{-1}_{\alpha,\gamma}(\alpha_1\wedge\delta_1),D^{-1}_{\alpha,\gamma}(-(\gamma_1\wedge\beta_1))]{\subseteq}[0,1]{,}\\
		\rho^*(1)&\in[D^{-1}_{\beta,\delta}(\gamma_1\wedge\beta_1),D^{-1}_{\beta,\delta}(-(\alpha_1\wedge\delta_1))]{=}[0,1],
	\end{align*} 
	{Where we remark that the equality on the second line above is due to $ \alpha_1\wedge\delta_1=\delta_1 $ and $ \gamma_1\wedge\beta_1=\beta_1 $.} To study $ \rho^*(1)-\rho^*(0)=D_{\beta,\delta}\rho^*(1) $, 
	let us first define the (solution) map $ \rho^*(1)\mapsto \phi_{\alpha,\gamma,\beta,\delta}(\rho^*(1)) $ where $ \phi_{\alpha,\gamma,\beta,\delta}(\rho^*(1)):=\rho^*(0) $ is the (unique for fixed $ \rho^*(1) $) solution (function of $ \rho^*(1) $) to the equation $ -D_{\alpha,\gamma}\rho^*(0)=D_{\beta,\delta}\rho^*(1) $. Moreover, define the function $ T $ as 
	\begin{align*}
		T:[D^{-1}_{\beta,\delta}(\gamma_1\wedge\beta_1),D^{-1}_{\beta,\delta}(-(\alpha_1\wedge\delta_1))]\to \mathbb{R},\quad u\mapsto T(u)=u-\phi_{\alpha,\gamma,\beta,\delta}(u)-D_{\beta,\delta}u.
	\end{align*} 
	{From our hypotesis on the parameters on the beginning of the proof, it is clear that $ Dom(T)=[0,1] $.} 	
	We claim that $ T $ is (Lipschitz) continuous and monotone increasing on $ [0,1] $. Assuming this, {we have}
	\begin{align*}
		-(\phi_{\alpha,\gamma,\beta,\delta}(0)+\beta_1)=T(0)\leq T(u)\leq T(1)=1+\delta_1-\phi_{\alpha,\gamma,\beta,\delta}(1)
	\end{align*} and thus there exists a  unique $ \rho^*(1) $ such that $ T\circ\rho^*(1)=0 $, and we are done. To prove the claim, consider $ g,h\in[0,1] $ and shorten $ g^*:=\phi_{\alpha,\gamma,\beta,\delta}(g) $ and $ h^*:=\phi_{\alpha,\gamma,\beta,\delta}(h) $. Then,
	\begin{align}\label{eq:T}
			T(g)-T(h)=g-h-(g^*-h^*)+(g-h)V_{\beta,\delta}(g,h).
	\end{align} 
	We now relate $ (g^*-h^*) $ and $ (g-h) $. From the definition of $ \phi $ we have that $ D_{\alpha,\gamma}g^*+D_{\beta,\delta}g=0 $ and $ D_{\alpha,\gamma}h^*+D_{\beta,\delta}h=0 $. Subtracting the two previous equations we observe that
	\begin{align*}
	(g^*-h^*)V_{\alpha,\gamma}(g^*,h^*)=-(g-h)V_{\beta,\delta}(g,h).
	\end{align*}
	In this way, replacing in \eqref{eq:T} the expression for $ g^*-h^* $ given in last display we arrive at 
	\begin{align*}
	T(g)-T(h)=(g-h)\left(
	1+V_{\beta,\delta}(g,h)+\frac{V_{\beta,\delta}(g,h)}{V_{\alpha,\gamma}(g^*,h^*)}
	\right)
	\end{align*}
	and since from Lemma \ref{lem:D1-D2} there are universal constants such that the  $ V_{\cdot,\cdot} $ terms above are bounded from below, and since $ g,h,g^*,h^*\in[0,1] $ they are also bounded from above, the claim is shown.\par
	{For the case $ \alpha_1\wedge\delta_1=\alpha_1 $ and $ \gamma_1\wedge\beta_1=\gamma_1 $, the proof is completely analogous. Observe that for every $ \rho^*(0) $ fixed, there exists  a unique solution for $ \rho^*(1) $ of the equation $ -D_{\alpha,\gamma}\rho^*(0)=D_{\beta, \delta}\rho^*(1) $. Following the same argument as on the previous case, we now have 
	\begin{align*}
	\rho^*(0)&\in[D^{-1}_{\alpha,\gamma}(\alpha_1\wedge\delta_1),D^{-1}_{\alpha,\gamma}(-(\gamma_1\wedge\beta_1))]=[0,1],\\
	\rho^*(1)&\in[D^{-1}_{\beta,\delta}(\gamma_1\wedge\beta_1),D^{-1}_{\beta,\delta}(-(\alpha_1\wedge\delta_1))]\subseteq[0,1],
	\end{align*} 	
	and we set to show that $ \rho^*(1)-\rho^*(0)=-D_{\alpha,\gamma}\rho^*(0) $ have one solution $ \rho^*(0) $, instead of $ \rho^*(1) $. For that, we now define the map $ \rho^*(0)\mapsto\phi_{\alpha,\gamma,\beta,\delta}(\rho^*(0)) $ where $ \phi_{\alpha,\gamma,\beta,\delta}(\rho^*(0))=:\rho^*(1) $ is the (unique for fixed $ \rho^*(0) $) solution (function of $ \rho^*(0) $) to the equation $ -D_{\alpha,\gamma}\rho^*(0)=D_{\beta,\delta}\rho^*(1) $, and define the map $ [0,1]\ni u \mapsto T(u)=u-\phi_{\alpha,\gamma,\beta,\delta}(u)-D_{\alpha,\gamma}u\in\mathbb{R}. $ Proceeding as on the previous case we have that $ T $ is Lipschitz continuous monotone increasing on $ [0,1] $ and exists a unique $ \rho^*(0) $ such that $ T\circ\rho^*(0)=0 $. 
	  }

\subsection{Uniqueness and convergence of \eqref{eq:Robin_equation}  to stationary solutions}
\label{appendix:convergence}
In this section, we assume \eqref{H2}, and we show the convergence of the weak solution of \eqref{eq:Robin_equation} to the unique stationary solution, denoted here by $ \rho^* $, investigated in the previous section. Existence and uniqueness of such a stationary solution is proved in Section \ref{sec:statio_uniq}. The main difficulty on showing this result lies on the fact that the weak solution is not regular enough w.r.t. time in order to have an integration by parts formula as we do w.r.t. space. To solve this issue, our approach is to relate the weak formulation and the mild formulation of \eqref{eq:Robin_equation}. We then show, following \cite{LT}, that mild and weak solutions are equivalent \textit{in some sense}, which indirectly gives us a regular enough version of the weak solution to then proceed with energy estimates. With a few adjustments from \cite{dptv} (Section $ 2.3 $), we first define the notion of mild solutions of \eqref{eq:Robin_equation}.

\begin{definition}[Mild solution of \eqref{eq:Robin_equation}]
	\label{def:mild_sol_ Robin}
	We call \emph{mild solution of } \eqref{eq:Robin_equation} any function $\rho:[0,T]\times[0,1]\to[0,1]$ satisfying $ M(\rho,t):=\rho_t-S\rho_t=0 $, with
\begin{align*}
	S\rho_t(u)=\int_0^1P_t(u,v)f_0(v)dv+
	\int_0^t \Big\{
	P_{t-s}(u,0)(D_{\alpha,\gamma}\rho_s)(0)+
	P_{t-s}(u,1)(D_{\beta,\delta}\rho_s)(1)\Big\}ds=0,
\end{align*}	
	where $P_t(u,v)=\sum_{w\in\Psi^{-1}(v)}\Phi_t(u,w)$ is the density kernel generated by the Laplacian $ \partial^2_u $ on $ [0,1] $ with reflecting Neumann boundary conditions, related to the heat kernel 
	\begin{align}\label{eq:heat_kernel}
		\Phi_t(u,w)=\frac{1}{(4\pi t)^{1/2}}e^{-\frac{(u-w)^2}{4t}},
	\end{align}
	by the reflection map $ \psi:\mathbb{R}\to[0,1] $ defined as
	\begin{align*}
	\psi(u+k)=
	\begin{cases}
	u,& u\in[0,1], k \text{ even},\\
	1-u,& u\in[0,1], k \text{ odd},
	\end{cases}
	\end{align*}
	extended to $ \mathbb{R} $ by the symmetry $ \psi(v)=\psi(-v) $, for $ v\in\mathbb{R} $.
\end{definition}
\begin{rem}\label{rem:regularity_S}
	Observe that fixed $ u\in[0,1] $, $ S\rho_t(u) $ is differentiable w.r.t. time, and given a smooth initial data $ f_0 $, we have that $ S\rho_t\in C^\infty(0,1) $. Moreover, there exists the limits $ \lim_{u\to 0} \frac{d^n}{du^n} S\rho_t(u)$ and $\lim_{u\to1}\frac{d^n}{du^n} S\rho_t(u) $ for any $ n\in\mathbb{N} $, and for any $ t\in[0,T] $ we have $ S\rho_t\in C[0,1] $. 
\end{rem}
Following \cite{LT}, with some adaptations to account for the fact that $ P_t $ involves here Neumann boundary conditions, we show below that if $ \rho $ is a weak solution of \eqref{eq:Robin_equation}, then $ \rho_t-S\rho_t=0 $ a.e. From the previous remark, $ S\rho $ is regular enough to satisfy $ F(\rho,S\rho,t)=0 $. Moreover, letting $ \rho^* $ be the stationary solution, as mentioned in Theorem \ref{thm:hydrostatics}, from simple energy estimates we can see that $ F(S\rho,S\rho,t)-F(\rho^*,S\rho,t)=0\implies \norm{S\rho_t-\rho^*}_{L^2}=O(e^{-Ct}) $ for positive constant $ C $, which implies weak convergence to the stationary profile. In this way, we have that $ \pi_t\to \pi^* $ in $ \mathcal{M}^+ $ (endowed with the weak topology), since $ \pi_t(du)=\rho_t(u)du, $ with $ \rho_t\xrightarrow{w} \rho^* $ and $\rho^*(u)du=:\pi^*(du) $, and thus we can apply Proposition \ref{hsconcentration}.

\begin{prop}\label{W=M}
If a function $\rho:[0,T]\to[0,1]$ is a weak solution  in the sense of Definition \ref{def:weak_sol_ Robin}, where it satisfies $ F(\rho,G,t)=0 $ for any $ G\in C^{1,2}([0,T]\times[0,1]) $, then $ \rho $ also satisfies $ M(\rho,t)=0 $ a.e. $\forall t>0$. Moreover, if $ \rho:[0,T]\times[0,1] $ is a function satisfying $ \inner{M(\rho,t),G}=0 $ for any $ G\in C^{1,2}([0,T]\times[0,1]) $, then we have $ F(S\rho,G,t)=0 $.		

\end{prop}
\begin{proof}
Let us fix $ g\in C^2([0,1]) $ and $ T\geq t\geq0 $ and define $ G^\eps: [0,T]\times[0,1]\to [0,1] $ as
\begin{align*}
G_s^\eps(u)=
(P_{t-s}g)(u)1_{\{s\in[0,t]\}}+
g(u)\frac{t+\eps-s}{\eps}1_{\{s\in [t,t+\eps]\}}+
(1-1_{\{s\in[t+\eps,T]\}}). 
\end{align*}
Above, for $  t\geq 0$ and $ u\in[0,1] $,  $ (P_tf)(u)=\int_0^1P_t(u,v)f(v)dv\equiv \inner{P_t(u,\cdot),f} .$  Assume now that $\rho$ is a weak solution  in the sense of Definition \ref{def:weak_sol_ Robin}. Recalling that for fixed $ v $,  $ P_t(u,v) $ satisfies the heat equation, with Neumann boundary conditions, taking $ G^\eps\in C^{1,2}([0,T]\times[0,1]) $ as our test function we have
\begin{align*}
F(\rho,G^\eps,T)&=-\inner{f_0,P_tg}
-\int_0^t
\Big\{
(D_{\beta,\delta}\rho_s)(1)(P_{t-s}g)(1)+(D_{\alpha,\gamma}\rho_s)(0)(P_{t-s}g)(0)
\Big\}ds+\int_t^{t+\eps}\inner{\rho_s,g}\eps^{-1}ds\\
-&\int_t^{t+\eps}
\Big\{
\inner{\rho_s,\partial^2_u g}
+(D_{\beta,\delta}\rho_s)(1)g
+(D_{\alpha,\gamma}\rho_s)(0)g
+\rho_s(1)\partial_u g(1)
+\rho_s(0)\partial_u g(0)
\Big\}\eps^{-1}(t+\eps-s)ds.
\end{align*}
Letting $ \eps\to0 $, $ G^\eps $ converges in the sup norm to $ G_s(u):=P_{t-s}g(u) $. As a consequence,  we have that $ F(\rho,G^\eps,T) $  converges to
\begin{align*}
\inner{\rho_t,g}-\inner{P_tf_0,g}
-\int_0^t\inner{(D_{\beta,\delta}\rho_s)(1)P_{t-s}(\cdot,1),g}ds
-\int_0^t\inner{(D_{\alpha,\gamma}\rho_s)(0)P_{t-s}(\cdot,0),g}ds=0.
\end{align*}
Approximating $G^{\eps}$ by a sequence $ (G^{\eps}_k)_{k\geq1} $ in $ C^{1,2}([0,T]\times[0,1]) $ w.r.t. the  $ L^1 $ norm, and since $ g\in C^2([0,1]) $ is arbitrary, proves as wanted that $M(\rho,t)=0 $ a.e. 

{For the converse, as in \cite{LT} and  for a better exposition, we consider a test function $ g $ independent of time, and we remark that the extension to a time dependent  function } is completely analogous.  Let us assume that $ \rho:[0,T]\times[0,1]\to[0,1] $ satisfies $ \rho=S\rho $ \textit{a.e.} $\forall t$. 
Then, for any $ g\in C^{2}([0,1]) $ it must satisfy 
\begin{align*}
	\inner{\rho_t,g}=\inner{P_tf_0,g}+\int_0^t\inner{P_{t-s}(\cdot,1)(D_{\beta,\delta}\rho_s)(1)+P_{t-s}(\cdot,0)(D_{\alpha,\gamma}\rho_s)(0),g}ds.
\end{align*}
In particular, differentiating the expression above with respect to time, we have
\begin{align}\label{eq:time_diff_inner}
\begin{split}
	\frac{d}{dt}\inner{\rho_t,g}=&
	\inner{(\partial^2_u P_t)g,f_0}
	+D_{\beta,\delta}\rho_t(1)g(1)+D_{\alpha,\gamma}\rho_t(0)g(0)\\+&\int_0^t\inner{\partial_tP_{t-s}(\cdot,1)D_{\beta,\delta}\rho_s(1)+\partial_tP_{t-s}(\cdot,0)D_{\alpha,\gamma}\rho_s(0),g}ds.
\end{split}
\end{align}
We now integrate by parts $ (\partial^2_u P_t)g(u)$ twice, that is:
\begin{align*}
	(\partial^2_u P_t)g(u) =\int_0^1g(v) \partial_v^2 P_t(u,v)dv
	&=\int_0^1 P_t(u,v)\partial_v^2 g(v)dv 
	-P_t(u,1)\partial_u g(1) + P_t(u,0)\partial_u g(0).
\end{align*}
 Above,  we used the fact that $ P_t $ satisfies Neumann boundary conditions. Since by assumption $M(\rho,t)=0$ a.e., we have a.e. $P_tf_0=\rho_t-
	\int_0^t \Big\{
	P_{t-s}(\cdot,0)(D_{\alpha,\gamma}\rho_s)(0)+
	P_{t-s}(\cdot,1)(D_{\beta,\delta}\rho_s)(1)\Big\}ds$. In particular, replacing the expression above in \eqref{eq:time_diff_inner}, yields
\begin{align*}
\frac{d}{dt}\inner{\rho_t,g}
&=
\inner{\rho_t,\partial^2_u g}
+(D_{\beta,\delta}\rho_t)(1)g(1)+(D_{\alpha,\gamma}\rho_t)(0)g(0)
+\partial_u g(0)\rho_t(0)-\partial_u g(1)\rho_t(1)\\&\quad
-\int_0^t \inner{P_{t-s}(\cdot,0)D_{\alpha,\gamma}\rho_s(0)+P_{t-s}(\cdot,1)D_{\beta,\delta}\rho_s(1)g(1),\partial^2_u g} ds\\&\qquad
-\partial_u g(0)\int_0^t P_{t-s}(0,0)D_{\alpha,\gamma}\rho_s(0)+P_{t-s}(0,1)D_{\beta,\delta}\rho_s(1)ds\\&\quad\qquad
+\partial_u g(1)\int_0^tP_{t-s}(1,0)D_{\alpha,\gamma}\rho_s(0)+P_{t-s}(1,1)D_{\beta,\delta}\rho_s(1)ds\\&\qquad\qquad
+\int_0^t\inner{\partial_tP_{t-s}(\cdot,1)D_{\beta,\delta}\rho_s(1)+\partial_tP_{t-s}(\cdot,0)D_{\alpha,\gamma}\rho_s(0),g}ds.
\end{align*}
Integrating by parts the fourth line on the previous display twice, using the fact that $ P_t $ satisfies the heat equation with Neumann boundary conditions, yields
\begin{align*}
	\frac{d}{dt}\inner{\rho_t,g}=\inner{\rho_t,\partial^2_u g}
	+(D_{\beta,\delta}\rho_t)(1)g(1)+(D_{\alpha,\gamma}\rho_t)(0)g(0)
	+\partial_u g(0)\rho_t(0)-\partial_u g(1)\rho_t(1).
\end{align*}
Integrating in time this identity yields as wanted that $ F(\rho, g, t)=0$ for any $g\in C^2([0,1])$. The case when $g$ is replaced by a time-dependent function $G$ is similar, we omit it.
\end{proof}
We now show that any weak solution converges exponentially to the stationary solution  $\rho^*$.
\begin{prop}
\label{prop:statio_conv}
There exists a constant $C>0$ such that any weak solution $\rho$  of \eqref{eq:Robin_equation} satisfies 
\[\norm{\rho-\rho^*}^2_{L^2}\leq e^{-2Ct},\]
where $\rho^*$ is the unique stationary solution to \eqref{eq:Robin_equation}.
\end{prop}
\begin{proof}
 From Lemma \ref{lem:D1-D2} we know that 
$ D_{\alpha,\gamma}\rho_t(0)-D_{\alpha,\gamma}\rho^*(0)=-(\rho_t(0)-\rho^*(0))V_{\alpha,\gamma}(\rho_t,\rho^*)(0,0). $ Let us write $ V(0,t):=V_{\alpha,\gamma}(\rho_t,\rho^*)(0,0) $ and $ V(1,t)=V_{\beta,\delta}(\rho_t,\rho^*)(1,1)$.
Both $ S\rho $ and $ \rho^* $ satisfy the weak formulation and from Remark \ref{rem:regularity_S} letting $ w_t(u):=S\rho_t(u)-\rho^*(u) $ we have $ F(S\rho,w,t)-F(\rho^*,w,t)=0 $, which rewrites as
\begin{multline*}
\inner{w_{t}, w_{t}}=  \inner{w_0,w_0}+
\int_0^t\inner{w_{s},\Big( \partial^2_u + \partial_s\Big) w_{s}} ds  
+	\int_0^t
w_{s}(0)\Big(
\partial_u  w_{s}(0)
-
w_s(0)
V_{\alpha,\gamma}(0,s)
\Big)
ds\\
-\int_0^t
w_{s}(1)
\Big(
\partial_u  w_{s}(1)
+
w_s(1)
V_{\beta,\delta}(1,s)
\Big)
ds.
\end{multline*}
Differentiating w.r.t. time we have
\begin{align*}
\frac{d}{dt}\inner{w_t,w_t}=\inner{w_t,(\partial_u^2+\partial_t)w_t}
+w_t(0)\left(\partial_uw_t(0)-w_t(0)V(0,t)\right)
-w_t(1)\left(
\partial_uw_t(1)+w_t(1)V(1,t)	
\right).
\end{align*}
Integrating by parts the first term on the r.h.s. of the previous display once in space, we obtain
\begin{equation}
\label{eq:clem11}
\frac{d}{dt}\inner{w_t,w_t}=\inner{w_t,\partial_tw_t}
-\inner{\partial_uw_t,\partial_uw_t}
-(w_t(0))^2V(0,t)
-(w_t(1))^2V(1,t).
\end{equation}
Now note that $ \inner{w_t,\partial_tw_t}=\frac12\frac{d}{dt}\inner{w_t,w_t} $. {We now need to derive a Poincar\'e-type inequality. For $ v\in[0,1] $ we have, {from Cauchy-Schwarz inequality, }
	\begin{align*}
		\int_0^v(\partial_uw_t(u))^2du
		\geq \Big(\int_0^v\partial_u w_t(u)du\Big)^2
		=(w_t(v))^2+(w_t(0))^2-2w_t(0)w_t(v),
	\end{align*}
	which implies
	\begin{align*}
		\int_0^1(w_t(v))^2dv
		&\leq \int_0^1(\partial_uw_t(u))^2du-(w_t(0))^2+2w_t(0)\int_0^1w_t(v)dv.
	\end{align*}
	From Young and  { Cauchy-Schwarz inequalities,} for any $ A>0 $
	\begin{align*}
		w_t(0)\int_0^1w_t(v)dv
		\leq \frac{1}{2A}(w_t(0))^2+\frac{A}{2}\int_0^1(w_t(v))^2dv,
	\end{align*}
	and {plugging this inequality in the previous display, we obtain}
	\begin{align*}
		(1-A)\int_0^1(w_t(v))^2dv
		\leq \int_0^1(\partial_uw_t(u))^2du+\Big(\frac{1}{A}-1\Big)(w_t(0))^2.
	\end{align*}
	Note that the inequality above allow us to trade some control over the boundary by some control over the bulk. The goal is to find an optimal $ A $ that still allow us to have exponential decay, w.r.t. the $ L^2 $ norm, to the steady state. In this way, {we obtain from \eqref{eq:clem11}, $V(1,t)$ being non-negative,}
	\begin{align*}
		0&=\frac12 \frac{d}{dt}\inner{w_t,w_t}
		+\inner{\partial_uw_t,\partial_uw_t}
		+(w_t(0))^2V(0,t)
		+(w_t(1))^2V(1,t)\\
		&\geq
		\frac12 \frac{d}{dt}\inner{w_t,w_t}
		+\inner{\partial_uw_t,\partial_uw_t}
		+(w_t(0))^2V(0,t)\\
		&\geq
		\frac12 \frac{d}{dt}\inner{w_t,w_t}
		+(1-A)\inner{w_t,w_t}
		+(w_t(0))^2 (1-\frac{1}{A}+ V(0,t)).
	\end{align*}
Recall that from Lemma \ref{lem:D1-D2} we know that $ V(0,t),V(1,t) $ are larger than some positive constant. We now need to choose $ A $ such that $ 1-\frac{1}{A}+ V(0,t)\geq0 $ for all $ t $, with $ A<1 $. Letting $ V(0,t)\geq \underline{v}_K>0 $ for some constant $ \underline{v}_K $ as in Lemma \ref{lem:D1-D2}, 
	{and choosing  $A=\frac{1}{1+\underline{v}_K}<1,$ we conclude that 
	\begin{align*}
		0\geq \frac12\frac{d}{dt}\norm{w_t}^2_{L^2}+\frac{\underline{v}_K}{1+\underline{v}_K}\norm{w_t}^2_{L^2}.
	\end{align*}
	From Gronwall's inequality  and since $\norm{w_0}^2_{L^2}\leq 1$, we conclude that $\norm{w_t}^2_{L^2}\leq \exp{-2\frac{\underline{v}_K}{1+\underline{v}_K}t}.$} Recalling that $ S\rho=\rho $ a.s. ends the proof.}

\end{proof}

\appendix
\section{Replacement Lemmas}
\label{appendix:replacement}
In this section we prove the replacements lemmas that are needed along the arguments presented above. We start by obtaining an estimate relating  the Dirichlet form and the carr\'e  du champ operator for this model.   As above, for simplicity of the presentation,  we state and prove the results for the case $K=2$, but the extension to the general case is completely analogous.   
\subsection{Dirichlet forms}

Let $\rho:[0,1]\rightarrow [0,1]$ be a measurable  profile and let $\nu^{N}_{\rho(\cdot)}$ be the Bernoulli product measure on $\Omega_{N}$ defined  by
\begin{equation}\label{eq:bern_prof}
\nu^{N}_{\rho(\cdot)}(\eta:\eta(x)=1)= \rho(\tfrac{x}{N}).
\end{equation}
For a probability measure $\mu$ on $\Omega_{N}$ and a density $f:\Omega_{N} \rightarrow \mathbb{R}$ with respect to $\mu$, the Dirichlet form  is defined as 

\begin{equation}\label{eq:dir_form}
\langle f, -\mc L_N f\rangle_\mu= \langle f, -\mc L_{N,0}f \rangle_{\mu}+{\tfrac{1}{N^\theta}}\langle f, -\mc L_{N,b}f \rangle_{\mu},
\end{equation}
and the  carr\'e du champ is defined by:
\begin{equation}\label{eq:false_dir}
D_N(\sqrt{f}, \mu):=D_{N,0}(\sqrt{f},\mu)+ {\tfrac{1}{N^\theta}}D_{N,b}(\sqrt{f},\mu),
\end{equation}
where  
\begin{equation*}
D_{N,0}(\sqrt{f},\mu)\;:=\;\sum_{x=1}^{n-2} \int_{\Omega_N}\left[ \sqrt{f(\eta^{x,x+1})}- \sqrt{f(\eta)}\right]^{2} d\mu,
\end{equation*}
\begin{equation}\label{eq:dn_b_L}
\begin{split}
D_{N,\pm}(\sqrt{f},\mu)=&\sum_{x\in I_\pm^K}\int
c_x^\pm(\eta)
\big[\sqrt{f(\eta^{x})}-\sqrt{f(\eta)}\big]^{2}\,d\mu,
\end{split}
\end{equation}
where we recall the rates $ c^\pm_x $ defined in \eqref{rates:boundary}
, and $ D_{N,b}=D_{N,-}+D_{N,+} $.
We claim that  for $\theta\geq 1$ and for  $\rho:[0,1]\rightarrow [0,1]$ a constant profile equal to, for example, $\alpha$, the following bound holds
\begin{equation}\label{eq:error_DF}
\begin{split}
\langle {\mc L}_N\sqrt{f},\sqrt{f} \rangle_{\nu^N_\alpha} &\lesssim  -D_N(\sqrt{f},\nu^N_\alpha) + O(\tfrac{1}{ N}).
\end{split}
\end{equation}
From Lemma 5.1 and Lemma 5.2 of  \cite{bmns} it is only necessary to control the contribution from the non-linear part of the boundary dynamics. To do that, it is enough to apply Lemma 5.1 of  \cite{BGJO} and the result follows. We leave these computations to the reader.
{\begin{rem}[On the bound of the Dirichlet form]
		Note that for any $ a,b $ we have the identity $ ab-b^2=-\tfrac12 (a-b)^2+\tfrac12 (a^2-b^2) $. This implies directly that 
		\begin{align*}
		\inner{\mathcal{L}_N\sqrt{f},\sqrt{f}}_{\mu}=-\frac{1}{2}D_N(\sqrt{f},\mu)+\frac{1}{2}E_{\mu}\left[(\mathcal{L}_Nf)(\eta)\right]
		\end{align*}
		for any measure $ \mu $. If  $ \mu $ is the invariant measure of the system, then the last term on the right hand side of last display  vanishes.  For this model we have no information about the stationary measure and by taking $ \mu=\nu_\alpha^N $ this term is non zero and has to be controlled. For the bulk dynamics, when computing that term with respect to the $\nu_\alpha^N$, it clearly vanishes, but the same is not true for the boundary dynamics. Since we are dealing with the case $\theta\geq 1$ we are reduced to bound $ N^{-\theta}E_{\nu_\alpha^N}[(\mathcal{L}_{N,b}f)(\eta)] $ and in this case the bounds provided by Lemma 5.1 of  \cite{BGJO} are enough. Nevertheless, in the case $\theta<1$ these bounds do not work any more and a new argument is needed. 
	\end{rem}}
\subsection{Replacement Lemmas}
\label{sec:repl_lemma}
We start this section by proving the next lemma which is the basis for the  replacement lemmas that are presented next.
\begin{lem}\label{lem:useful_0}
	Let $x<y\in\Lambda_N$ and let {$\varphi:\Omega_N\to\Omega_N$ be a}  bounded function which satisfies $\varphi(\eta)=\varphi(\eta^{z,z+1})$ for any $z=x,\cdots, y-1$.  For any density $f$ with respect to $\nu_{\alpha}$ and any positive constant $A$, it holds that 
	$$\left| \langle\varphi(\eta)( \eta(x)-\eta(y)), f \rangle_{\nu^N_\alpha} \right| \lesssim \tfrac{1}{A} D_N(\sqrt{f},\nu^N_\alpha) + A(y-x).$$
\end{lem}
\begin{proof}
	By summing and subtracting appropriate terms, we have that
	\begin{equation*}
	\begin{split}
	\vert \langle \varphi(\eta)(  \eta(x)-\eta(y)),f \rangle_{\nu^N_\alpha} \vert 
	&\leq  \dfrac{1}{2} \sum_{z=x}^{y-1}\left\vert \int\varphi(\eta)(  \eta(z)-\eta(z+1))[f(\eta)-f(\eta^{z,z+1})] \; d\nu^N_\alpha\right\vert \\ 
	&+ \dfrac{1}{2}\sum_{z=x}^{y-1} \left\vert \int \varphi(\eta)(  \eta(z)-\eta(z+1))[f(\eta)+f(\eta^{z,z+1})] \; d\nu^N_\alpha\right\vert.
	\end{split}
	\end{equation*}
	Note that since $\varphi$ satisfies $\varphi(\eta)=\varphi(\eta^{z,z+1})$ for any $z=x,\cdots, y-1$, by a change of variables, we conclude that the last term in the previous display is equal to zero. Now, we treat the remaining term. 
	Using the equality $(a-b)=(\sqrt a-\sqrt b)(\sqrt a+\sqrt b)$ and then Young's inequality, the first term at the {right-hand} side of last display is bounded from above by a constant times
	\begin{equation*}
	\begin{split}
	& \sum_{z=x}^{y-1}{\frac A4}\int (\varphi(\eta)(\eta(z)-\eta(z+1)))^2\left( \sqrt{f(\eta^{z,z+1})}+ \sqrt{f(\eta)}\right)^{2}  d\nu^N_\alpha+\dfrac{1}{{4A}}D_N(\sqrt{f},\nu^N_\alpha).
	\end{split}
	\end{equation*} 
	The fact that $\varphi$ is bounded, $|\eta(x)|\leq 1$ and $f$ is a density, the {integral in }  last expression is bounded from above by a constant. This ends the proof.
\end{proof}

We are now able to show the first Replacement Lemma.
\begin{lem}\label{replacement lemma_0}\label{point}
	Fix $x,y\in\Lambda_N$ such that $|x-y|=o(N)$. Let {$\varphi:\Omega_N\to\Omega_N$ be a}  bounded  function which satisfies $\varphi(\eta)=\varphi(\eta^{z,z+1})$ for any $z=x,\cdots, y-1$. For any $t\in [0,T]$ we have that
	\begin{equation}\label{replacement}
	\begin{split}
	&\displaystyle \limsup_{N\to +\infty} \mathbb{E}_{\mu_{N}}\left[\, \left| \displaystyle\int_{0}^{t}\varphi(\eta_{sN^2})(\eta_{sN^2}(x)- \eta_{sN^{2}}(y))ds \right| \,\right] = 0.
	\end{split}
	\end{equation}
\end{lem}
\begin{proof}
{The starting point in the proof is to change from the measure $\mu_N$ to a suitable measure, which for our purposes is the Bernoulli product measure   $\nu_{\rho(\cdot)}^N$
with a constant profile $\rho(\cdot)=\alpha\in(0,1)$. By the explicit formula for the entropy, it holds 
\begin{equation*}
\begin{split}
    H(\mu_N|\nu_{\alpha}^N)&=\sum_{\eta\in\Omega_N} \mu_N(\eta)\log\bigg(\frac{\mu_N(\eta)}{\nu_{\rho(\cdot)}^N(\eta)}\bigg)     
    \leq N \log\bigg(\frac{1}{C_{\alpha}}\bigg)\sum_{\eta\in\Omega_N} \mu_N(\eta)=C_{\alpha} N.
\end{split}
\end{equation*}
Therefore, by the entropy inequality and Jensen's inequality, for any $B>0$ the expectation in the statement of the lemma  is bounded by 
\begin{equation}\label{after jensen}
\frac {C_\alpha} {B} + \frac {1} {NB} \log \mathbb E_{\nu_{\alpha}^N}\bigg[ e^{|\int_0^t BN\varphi(\eta_{sN^2})(\eta_{sN^2}(x)- \eta_{sN^{2}}(y))\,ds| } \bigg].
\end{equation}
Since $e^{|x|}\leq e^x+e^{-x}$ and 
\begin{equation*}
    \underset{N\to\infty}\limsup\, \tfrac{1}{N}\log (a_N+b_N)\leq \max\big\{\underset{N\to\infty}\limsup\, \tfrac{1}{N}\log (a_N),    \underset{N\to\infty}\limsup\, \tfrac{1}{N}\log (b_N)  \big\},
\end{equation*}
we can remove the absolute value from \eqref{after jensen}. By Feynman-Kac's formula (see Lemma 7.3 in \cite{bmns}), \eqref{after jensen} is bounded by 
\begin{equation*}
  \frac {C_{\alpha}} {B} +   t\sup_{f} \Big \{ |\langle \varphi(\eta)(\eta(x)-\eta(y)), f \rangle_{\nu^N_\alpha} |+ \tfrac{N}{B} \langle \mc L_{N}\sqrt{f},\sqrt{f} \rangle_{\nu^N_\alpha} \Big \},
	\end{equation*}
	where the supremum above is  over  densities $f$ with respect to $\nu_\alpha^N$. By Lemma \ref{lem:useful_0} with the choice $A=\tfrac BN$ we have that 
	$$\left| \langle\varphi(\eta)( \eta(x)-\eta(y)), f \rangle_{\nu^N_\alpha} \right| \lesssim \tfrac{N}{B} D_{N}(\sqrt{f},\nu^N_\alpha) + \tfrac{B}{N}{|y-x|}{.} $$
	From \eqref{eq:error_DF} and the inequality above, the term on the right-hand side of \eqref{after jensen}, is bounded from above by
$
	\tfrac{B}{N}{|y-x|} +{\tfrac {1}{B}}.
$
	Taking $N \to \infty$ and then $B\to+\infty$ we are done.}
\end{proof}
\begin{lem}[{Replacement Lemma}]\label{replacement lemma_1}
	Let {$\psi:\Omega_N\to\Omega_N$ be a}  bounded  function which satisfies $\psi(\eta)=\psi(\eta^{z,z+1})$ for any $z=x+1,\cdots, x+\varepsilon N-1$.  For any $t\in [0,T]$ and $x\in\Lambda_N$ such that $x\in\{1,\cdots, N-\varepsilon N-2\}$  we have that
	\begin{equation}\label{replacement_1}
	\begin{split}
	&{\limsup_{\varepsilon\to 0}}\displaystyle\limsup_{N\to +\infty} \mathbb{E}_{\mu_{N}}\left[\, \left| \displaystyle\int_{0}^{t}\psi(\eta_{sN^2})(\eta_{sN^2}(x)- \overrightarrow{\eta}^{\varepsilon N}_{sN^2}(x))ds \right| \,\right] = 0.
	\end{split}
	\end{equation}
	Note that for $x\in\Lambda_N$ such that  $x\in\{N-\varepsilon N-1, N-1\}$ the previous result is also true, but we replace in the previous expectation $\overrightarrow{\eta}^{\varepsilon N}_{sN^2}(x)$ by $\overleftarrow{\eta}^{\varepsilon N}_{sN^2}(x)$, {where both averages were defined in \eqref{eq:emp_aver}}.
\end{lem}
\begin{proof}
	We present the proof for the case when $x\in\{1,\cdots, N-\varepsilon N-2\}$ but we note that the other case is completely analogous. By applying the same arguments as in the proof of the previous lemma {and by changing to the Bernoulli product measure $\nu_\alpha^N$ with $\alpha\in(0,1)$}, we can bound from above the previous expectation by 
	\begin{equation}\label{eq:FK_comp2}
	\begin{split}
	\tfrac{{C_{\alpha}}}{B} + t\sup_{f} \Big \{{ |}\langle \psi(\eta)(\eta(x)-\overrightarrow{\eta}^{\varepsilon N}(x)), f \rangle_{\nu^N_\alpha} {|}+ \tfrac{N}{B} \langle \mc L_{N}\sqrt{f},\sqrt{f} \rangle_{\nu^N_\alpha} \Big \}.
	\end{split}
	\end{equation}
	where $B$ is a positive constant. The supremum above is  over  densities $f$ with respect to {$\nu_\alpha^N$}. 
	The first term in the supremum above can be {bounded by}  
	\begin{equation*}
	\frac{1}{\varepsilon N}\sum_{y=x+1}^{x+\varepsilon N}{|}\langle \psi(\eta)(\eta(x)-\eta(y)), f \rangle_{\nu^N_\alpha} {|.}
	\end{equation*}
	By Lemma \ref{lem:useful_0} with the choice $A=\tfrac BN$  and from \eqref{eq:error_DF}, the term on the right-hand side of \eqref{eq:FK_comp2}, is bounded from above by
$
	{B\varepsilon} +\tfrac {1}{N}.
	$
	Taking $N \to \infty$, ${\varepsilon\to 0} $ and then $B\to+\infty$ we are done.
\end{proof}
Now we state the  replacement lemma that we need when the process is speeded up in the subdiffusive time scale.
\begin{cor}\label{replacement lemma_mass}
	 Recall from \eqref{eq:defmt} the definition of the mass of the system $m^N_t$ at the subdiffusive time scale $tN^{1+\theta}$. For any $\theta>1$ and  $t\in [0,T]$ and $x\neq z\in\Lambda_N$ we have that
	\begin{equation}\label{replacement_mass}
	\begin{split}
	&\displaystyle\limsup_{N\to +\infty} \mathbb{E}_{\mu_{N}}\left[\, \left| \displaystyle\int_{0}^{t}\eta_{sN^{1+\theta}}(z)(\eta_{sN^{1+\theta}}(x)-m_s^N )ds \right| \,\right] = 0.
	\end{split}
	\end{equation}
\end{cor}
The proof follows exactly the same strategy as above, the only difference being that when we use Lemma \ref{lem:useful_0} the function $\varphi(\eta)=\eta(z)$ is not invariant under the exchanges in the bulk. Nevertheless, by  observing that the integrand function above can be written as 
{\begin{equation}\label{replacement_mass}
	\begin{split}
\eta(z)(\eta(x)-\langle\pi^N,1\rangle )&=\frac{\eta(z)}{N-1}\sum_{y\neq x}(\eta(x)-\eta(y))\\&=\frac{\eta(z)}{N-1}(\eta(x)-\eta(z))+\frac{\eta(z)}{N-1}\sum_{y\neq x,z}{(\eta(x)-\eta(y))}
	\end{split}
	\end{equation}}
	and {  thanks to the exclusion rule the first term in the last line of  last display vanishes as $N\to+\infty$,  and to finish the proof it is enough to estimate the second term in the last line of last display. To finish the proof,  we distinguish two cases: either $x<z$ or $x>z$. We do the proof for the  case $x<z$, but the other one is completely analogous. Observe that 
\begin{equation*}
	\begin{split}
\frac{\eta(z)}{N-1}\sum_{y\neq x,z}{(\eta(x)-\eta(y))}=\frac{\eta(z)}{N-1}\sum_{y=1}^{z-1}{(\eta(x)-\eta(y))}+\frac{\eta(z)}{N-1}\sum_{y=z+1}^{N-1}{(\eta(x)-\eta(y))}.
	\end{split}
	\end{equation*}
	Now we explain how to estimate each one of the terms in the previous display. We present the proof for the first term but the second one is analogous. 
	Therefore, we have to estimate 
		\begin{equation*}
	\begin{split}
	& \mathbb{E}_{\mu_{N}}\left[\, \left| \displaystyle\int_{0}^{t}\frac{\eta_{sN^{1+\theta}}(z)}{N-1}\sum_{y=1}^{z-1}(\eta_{sN^{1+\theta}}(x)-\eta_{sN^{1+\theta}}(y))ds \right| \,\right].
	\end{split}
	\end{equation*}
	Now, we mimic the proof of the previous lemma. By following the first part of the proof, last expectation is bounded from above by 
	\begin{equation*}
	\begin{split}
	\tfrac{{C_{\alpha}}}{B} + t\sup_{f} \Big \{{ |}\langle \frac{\eta(z)}{N-1}\sum_{y=1}^{z-1}{(\eta(x)-\eta(y))}, f \rangle_{\nu^N_\alpha} {|}+ \tfrac{N^\theta}{B} \langle \mc L_{N}\sqrt{f},\sqrt{f} \rangle_{\nu^N_\alpha} \Big \},
	\end{split}
	\end{equation*}
	where $B$ is a positive constant. The supremum above is  over  densities $f$ with respect to {$\nu_\alpha^N$}. Now, repeat the proof of Lemma  \ref{lem:useful_0} and the previous lemma and to conclude,  make the choice  $A=BN^{-\theta}$. We leave the details to the reader.}

\section{Energy Estimate}
\label{appendix:energy}

Now we  prove that the density $\rho(t,u)$ belongs to the space $L^{2}(0,T;\mathcal{H}^{1})$, see Definition \ref{Def. Sobolev space}.  {Define} the  linear functional $\ell_{\rho}$ defined in $C^{0,1}_{c}([0,T]\times (0,1))$ by 
$$ \ell_{\rho}(G) = \int^{T}_{0}\int^{1}_{0}\partial_uG_{s}(u)\rho(s,u) \, du ds = \int^{T}_{0}\int^{1}_{0}\partial_uG_{s}(u)\,  \pi(s,du) ds.$$

\begin{prop} 
	\label{prop:energy_estimate}There exist positive constants $C$ and $c$ such that 
	$$\mathbb{E}\Big[ \sup _{G\in C^{0,1}_{c}([0,T]\times (0,1))}\Big\{ \ell_{\rho}(G) - c \Vert G \Vert_{2}^{2}\Big\}\Big]\leq C < \infty.$$
	Above $\Vert G \Vert_{2}$ denotes the norm of a function $G \in L^{2}([0,T]\times (0,1)).$
\end{prop}

{Before proving this result, we state and prove an energy estimate for the macroscopic density.
\begin{cor}\label{cor:EnergyEstimate}
Any limit point $\mathbb{Q}$ of the sequence $(\mathbb{Q}_N)_{N\geq 1}$ satisfies
\[\mathbb{Q}\left(\pi_\cdot\in \mathcal{D}([0,T],\mathcal{M}), \;\;\pi_t:=\rho_t(u) du, \;\;\rho\in L^2(0;T,\mathcal{H}^1)\right)=1.\]
We denote $R_T$ the event above. 
\end{cor}
\begin{proof}[Proof of Corollary \ref{cor:EnergyEstimate}]
We first note that because of the exclusion between particles,  every limit point  $\mathbb{Q}$
is concentrated on trajectories of measure that are  absolutely continuous with respect to the Lebesgue measure (see e.g. \cite{KL}, p.57, last paragraph for more details).

From Proposition \ref{prop:energy_estimate}, $\ell_{\rho}$ is $\bb Q$-almost surely continuous and therefore we can extend this linear functional to $L^{2}([0,T]\times (0,1))$. As a consequence of the Riesz's Representation Theorem  there exists $H \in L^{2}([0,T]\times (0,1))$ such that
$$\ell _{\rho}(G) = -\int^{T}_{0}\int^{1}_{0} G_{s}(u)H_{s}(u)du ds$$ for all $G \in C^{0,1}_{c}([0,T]\times (0,1))$. From this we conclude that $\rho \in L^{2}(0,T;\mathcal{H}^{1})$.
\end{proof}
}
\begin{proof}[Proof of Proposition \ref{prop:energy_estimate}]
	By density and by the Monotone Convergence Theorem it is enough to prove that for a countable dense subset $\lbrace  G_{m}\rbrace_{m \in \mathbb{N}}$ on $C_{c}^{0,2}([0,T]\times (0,1))$ it holds that
	$$\mathbb{E}\left[ \max _{k \leq m}\lbrace \ell_{\rho}(G^{k}) - c  \Vert G^{k} \Vert_{2}^{2}\rbrace \right]\leq C_{0},$$
	for any $m$ and for $C_{0}$ independent of $m$. 
	Note that the function that associates to a trajectory $\pi_\cdot \in \mc D([0,T], \mc M^+)$ the number 
	$\max_{k\leq m}\left\{ 
	\ell_{\rho}(G^{k})- c\Vert G^{k} \Vert_{2}^{2}
	\right\},$ 
	is  continuous and bounded  w.r.t.  the Skorohod topology of $ \mc D([0,T], \mc M^+)$ and for that reason, the expectation in the previous display is equal to the next limit
	\begin{eqnarray*}
		\lim _{N\rightarrow \infty} \EE _{\mu _{N}}\left[ \max_{k\leq m}\left\lbrace \int_{0}^{T}\dfrac{1}{N-1} 
		\sum_{x=1}^{N-1} \partial _{u}G_{s}^{k}(\tfrac{x}{N})\eta_{{sN^2}}(x)ds -c\Vert G^{k} \Vert_{2}^{2} \right\rbrace \right].
	\end{eqnarray*}
	By   entropy  and  Jensen's inequalities plus the fact that $e^{\max_{k\leq m} a_{k}}\leq \sum_{k=1}^{m}e^{a_{k}}$ the previous display is bounded from above by 
	\begin{eqnarray*}\nonumber
		&&C_{0}
		+ \dfrac{1}{N}\log \mathbb{E}_{\nu^N_\alpha}\left[ \sum _{k =1}^{m}{\exp
	\Big\{{\int _{0}^{T}\sum _{x \in \Lambda_{N}}\partial_{u}G_{s}^{k}(\tfrac{x}{N})\eta_{{sN^2}}(x) ds - cN\Vert G^{k} \Vert_{2}^{2} }\Big\}} \right],\\
	\end{eqnarray*}
	By linearity of the expectation, to treat the second term in the previous display it is enough  to bound the term 
	$$\limsup_{N \to \infty} \;  \dfrac{1}{N}\log \mathbb{E}_{\nu^N_\alpha}\left[ {\exp
	\Big\{\int _{0}^{T}\sum _{x \in \Lambda_{N}}\partial_{u}G_{s}(\tfrac{x}{N})\eta_{{sN^2}}(x) ds - cN\Vert G \Vert_{2}^{2} \Big\}}\right],$$
	for a fixed function $G\in C_{c}^{0,2}([0,T]\times (0,1))$, by a constant independent of $G$.
	By Feynman-Kac's formula, the expression inside the limsup is bounded from above by 
	\begin{equation} 
	\label{EE3}
	\int _{0}^{T}\sup _{f}\Big\{\dfrac{1}{N}\int_{\Omega_{N}}\sum _{x \in \Lambda_{N}}\partial_{u}G_{s}(\tfrac{x}{N})\eta(x) f(\eta)d{\nu^N_{\alpha} } - c\Vert G \Vert_{2}^{2}  + N \langle \mc L_{N}\sqrt{f},\sqrt{f}\rangle_{{\nu^N_{\alpha}}}\Big\}\, ds{,}
	\end{equation}
	where the supremum is carried over all the densities $f$ with respect to $\nu^N_{\alpha}$. Note that by a Taylor expansion on $G$, it is easy to see that we can replace its space derivative by the discrete gradient $\nabla^+_{N} G_{s}(\tfrac{x-1}{N})$ by paying an error of order $O(\frac 1N)$. Then, from a summation by parts, we obtain
	\begin{equation*}
	\int_{\Omega_{N}} \sum_{x=1}^{N-2} G_{s}(\tfrac{x}{N})(\eta(x)-\eta(x+1))f(\eta)d\nu^N_{\alpha}
	\end{equation*}
	By writing the previous term as one half of it plus one half of it and in one of the halves we swap the occupation variables $\eta(x)$ and  $\eta(x+1)$, for which the measure $\nu_\alpha$ is invariant,  the last display becomes equal to
	\begin{equation}
	\begin{split}
	\label{EE5}
	&\dfrac{1}{2}\int_{\Omega_{N}} \sum_{x=1}^{N-2} G_{s}(\tfrac{x}{N})(\eta(x)-\eta(x+1))(f(\eta)-f(\eta^{x,x+1}))d{\nu^N_{\alpha}}.
	\end{split}
	\end{equation}
	Repeating similar arguments to those used in the proof of Lemma \ref{lem:useful_0}, the last term is bounded from above by 
	\begin{eqnarray}\label{EE6}\nonumber
	&\dfrac{1}{4N}\int_{\Omega_{N}} \sum_{x=1}^{N-2} (G_{s}(\tfrac{x}{N}))^{2}(\sqrt{f(\eta)}+ \sqrt{ f(\eta^{x,x+1})})^{2}d\nu^N_{\alpha} +  \dfrac{1}{4N}\int_{\Omega_{N}} \sum_{x=1}^{N-2} (\sqrt{f(\eta)}- \sqrt{ f(\eta^{x,x+1})})^{2}d\nu^N_{\alpha} \\\nonumber
	&\leq  \dfrac{C }{N}\sum_{x\in \Lambda_{N}} (G_{s}(\tfrac{x}{N}))^{2}+ \dfrac{1}{4N} D_{0,N}(\sqrt{f},\nu^N_{\alpha}) 
	\end{eqnarray}
	for some $C>0$. 
	From \eqref{eq:error_DF} we get  that \eqref{EE3} is bounded from above by 
	\begin{equation*}
	C \int_0^T \Big[ 1 +  \dfrac{1}{N} \sum_{x\in \Lambda_{N}} (G_{s}(\tfrac{x}{N}))^{2}\Big] \, ds  \; -\;  c \| G\|_2^2 
	\end{equation*}
	plus an error of order $O(\tfrac 1N)$. Above
	$C$ is a positive constant independent of $G$.  Since $\dfrac{1}{N} \sum_{x\in \Lambda_{N}} (G_{s}(\tfrac{x}{N}))^{2}$ converges, as $N\to+\infty$, to $\|G\|_2^2$, then it is enough to  choose $c>C$ to conclude that 
	$$\limsup_{N\to \infty} \; \Big\{ C \int_0^T \Big[ 1 +  \dfrac{1}{N} \sum_{x\in \Lambda_{N}} (G_{s}(\tfrac{x}{N}))^{2}\Big] \, ds  \; -\;  c \| G\|_2^2  \Big\} \; \lesssim \; 1$$ and we are done.
\end{proof}

\section{Uniqueness of weak solutions of \eqref{eq:Robin_equation}}
\label{ap:uni_weak}

We start this section by recalling from  {Section 7.2 of }  \cite{BDGN}  the next two lemmas, which will be used in our proof. {The first one concerns uniqueness of the strong solutions of the heat equation with \emph{linear} Robin boundary conditions.}
\begin{lem} \label{lem:rob}
	For any $ t\in(0,T] $, the following problem with Robin boundary conditions
	\begin{align}\label{eq:lem_robin}
		\begin{cases}
			\partial_s\varphi(s,u)+a\partial^2_u\varphi(s,u)=\lambda \varphi(s,u),&(s,u)\in [0,t{)}\times(0,1),\\
			\partial_u \varphi(s,0)=b(s,0)\varphi(s,0), & s\in[0,t),\\
			\partial_u \varphi(s,1)=-b(s,1)\varphi(s,1), & s\in[0,t),\\
			\varphi(t,u)=h(u), & u\in(0,1),
		\end{cases}
	\end{align}
	with $ h\equiv h(u)\in C_0^2([0,1]) $ , 
	$ \lambda \geq0 $ , 
	$ 0<a\equiv a(u,t)\in C^{2,2}([0,T]\times[0,1]) $, and for $ u\in\{0,1\} $, 
	$ 0<b\equiv b(u,t)\in C^2[0,T] $, 
	has a unique solution $ \varphi\in C^{1,2}([0,t]\times [0,1]) $. Moreover, if $ h\in[0,1] $ then we have $ \forall (s,u)\in [0,t]\times [0,1] $:
	\begin{align*}
		0\leq \varphi(s,u)\leq e^{-\lambda(t-s)}.
	\end{align*}
\end{lem}
The second lemma is a technical regularization result on the coefficients $b(s,\cdot).$
\begin{lem}\label{lem:approx_b}
	Let $ 0\leq b $ be a bounded measurable function in $ [0,T] $, $ A=\{t\in[0,T]:b(t)>0\} $ and $ p\in[1,\infty) $. 
	Then there is a sequence $  (b_k)_{k\geq0} $ of positive functions in $ C^{\infty}[0,T] $ such that $ b_k\xrightarrow{k\to\infty}b $ in $ L^p([0,T]) $ and
	\begin{align*}
		\norm{\frac{b}{b_k}-1}_{L^{p}(A)}\xrightarrow{k\to\infty}0.
	\end{align*}
\end{lem}

For the proof of Lemma \ref{lem:uniqueness}, that is of uniqueness  of weak solutions  of \eqref{eq:Robin_equation}, we will follow Filo's method \cite{filo}, but mostly as presented in  Section 7.2 of \cite{BDGN}. The main idea is to choose a particular test function for the weak formulation satisfied by $ w:=\rho^{(1)}-\rho^{(2)}$, where $ \rho^{(1)} $ and $ \rho^{(2)} $ are two weak  solutions with the same initial data.
Although we do not have as much work to treat the bulk terms as in \cite{BDGN}, our main issue is the non linearity of the boundary conditions.

Recalling the weak formulation in \eqref{eq:Robin_integral} and Lemma \ref{lem:D1-D2}, since $$ D_{\lambda,\sigma}\rho_s^{(1)}(v)-D_{\lambda,\sigma}\rho_s^{(2)}(v)=-w_s(v) V_{\lambda,\sigma}(\rho_s^{(1)},\rho_s^{(2)})(v,v):=-w_s(v)V_{\lambda,\sigma}(v,s) $$ for $ v=0,1 $ and $ (\lambda,\sigma)=(\alpha,\gamma),(\beta,\delta) $, we have:
\begin{multline}\label{weak_w}
	\inner{w_{t}, G_{t}}=   
	\int_0^t\inner{w_{s},\Big( \partial^2_u + \partial_s\Big) G_{s}} ds  
	+	\int_0^t
	w_{s}(0)\Big(
	\partial_u  G_{s}(0)
	-
	G_s(0)
	V_{\alpha,\gamma}(0,s)
	\Big)
	ds\\
	-\int_0^t
	w_{s}(1)
	\Big(
	\partial_u  G_{s}(1)
	+
	G_s(1)
	V_{\beta,\delta}(1,s)
	\Big)
	ds.
	\end{multline}
Now we choose our test functions. Since $V_{\cdot,\cdot} $ does not have enough regularity, we have to overcome this problem by using Lemma \ref{lem:approx_b}. We focus on the left boundary, since for the right boundary the computations are analogous. Let $ A_0=\{t\in[0,T]: V_{\alpha,\gamma}(0,t)>0 \} $ (similarly, we define $ A_1 $ with respect to the right boundary). From Lemma \ref{lem:D1-D2} we have $ V_{\alpha,\gamma}(0,s)>0 $ and we may therefore exchange $ [0,t] $ by $ A_0 $ (resp. $ A_1 $) and apply Lemma \ref{lem:approx_b}. As a consequence of Lemma \ref{lem:approx_b},   for $ k$ large enough, there exists $ b_{k}(s,0) $ close to $  V_{\alpha,\gamma}(0,s)$ in $ L^p([0,T]) $ for $p\in [1,+\infty)$:
\begin{align*}
	\norm{\frac{ V_{\alpha,\gamma}(0,\cdot)}{b_k(\cdot,0)}-1}_{L^p(A_0)}\leq\eps
\end{align*}
for $ \eps>0$ and 
$ A_0=\{
s\in[0,t]:  V_{\alpha,\gamma}(0,s)>0
\} $. Now we choose the space of test functions as a sequence $ \varphi_k $, where  for each $k$, the function $ \varphi_k $ solves \eqref{eq:lem_robin} with $ b_k(\cdot,0)$ given above and  with $ \lambda=0 $.
From the boundary conditions of \eqref{eq:lem_robin}, second term in \eqref{weak_w} writes as
\begin{align}\label{eq:boundary_1}
	&\int_0^t
	\varphi_k(0,s)w_s(0)
	\left(
	b_k(0,s)- V_{\alpha,\gamma}(0,s)
	\right)
	ds.
\end{align}
Exchanging $ [0,t] $ by $ A_0 $, the last display can be bounded from above by
\begin{align*}
	\Big|\int_{A_0}
	\varphi_k(0,s)w_s(0)b_k(0,s)
	\left(
	1-\frac{ V_{\alpha,\gamma}(0,s)}{b_k(0,s)}
	\right)
	ds\Big|
	\leq 
	2\norm{b_k(0,s)}_{L^1(A_0)}\norm{\frac{ V_{\alpha,\gamma}(0,s)}{b_k(0,s)}-1}_{L^1(A_0)}\lesssim \eps,
\end{align*}
where we used that  $\varphi_k(\cdot)$ and $w(\cdot)$ are bounded functions.   For the right boundary the argument is completely analogous.

Now we treat the bulk term. From our choice of test function we have
\begin{align*}
	\int_0^t\inner{w_s,(\partial^2_u+\partial_s)\varphi_k}ds
	=\int_0^t\inner{w_s,(1-a)\partial^2_u \varphi_k(\cdot,s)}ds.
\end{align*}
Letting $ a=1 $, we thus have that $ \inner{w_t,\varphi_t}\lesssim \eps $. Since $ \varphi_t=h $, it is enough to take $ h\equiv h_k\in C_0^2([0,1]) $ such that $ h_k(\cdot)\xrightarrow{k\to\infty}1_{\{u\in[0,1]:w_t(u)>0 \}}(t,\cdot) $ in $ L^2([0,1]) $. The conclusion follows straightforwardly.
\begin{rem}
  We remark that for the model in \cite{dptv} the lower bound takes the form of the term $ x=K $ in the sum above, with $ \gamma=0 $. For $ K=1 $, that is, the case studied in \cite{bmns}, we have $ V=\rho^{(1)}+\rho^{(2)} $, and thus $ V=0\implies w=0 $.
\end{rem}


\begin{thebibliography}{99}

	\bibitem{bmns} 
	Baldasso, R., Menezes, O., Neumann, A., Souza, R. R.: Exclusion Process with Slow Boundary, Journal of Statistical Physics, Volume  167, no.  5, 1112--1142 (2017).
	
	\bibitem{BDGN}
	Bonorino, L., De Paula, R., Gon\c calves, P., Neumann, A.: Hydrodynamics for the porous medium model with slow reservoirs, arXiv:1904.10374 (2019).
	
	
	\bibitem{BGJO} 
	Bernardin, C.,  Gon\c{c}alves, P., Jim\'enez-Oviedo, B.:  Slow to fast infinitely extended reservoirs for the symmetric exclusion process with long jumps, Markov Processes and Related Fields, no. 25, 217--274 (2019).
	
	
	\bibitem{BGJO2} 
	Bernardin, C.,  Gon\c{c}alves, P., Jim\'enez-Oviedo, B.: A microscopic model for a one parameter class of fractional laplacians with Dirichlet boundary conditions, arXiv:1803.00792 (2018).


\bibitem{Derrida} B. Derrida: Non-equilibrium steady states: Fluctuations
  and large deviations of the density and of the current.
  J. Stat. Mech.  Theory Exp., P07023 (2007).
  
  
	\bibitem{DMP12}
	De Masi, A., Presutti, E., Tsagkarogiannis, D.,  Vares, M: Truncated correlations in the stirring process with births and deaths, Electronic Journal of Probability,  Volume 17, no. 6 (2012).
	
	
	\bibitem{dptv3}
	De Masi, A., Presutti, E., Tsagkarogiannis, D., Vares, M.: Non-equilibrium Stationary States in the Symmetric Simple Exclusion with Births and Deaths, Journal of Statistical Physics, Volume 147, no. 3, 519--528 (2012).
	
	\bibitem{dptv}
	De Masi, A., Presutti, E., Tsagkarogiannis, D., Vares, M.: Current Reservoirs in the Simple Exclusion Process, Journal of Statistical Physics, Volume 144, no. 3, 519--528 (2011).
	
	\bibitem{EGN2}
Erignoux, C., Gon\c calves, P. Nahum, G.: \emph{Hydrodynamics for  SSEP with non-reversible slow  boundary dynamics: Part II, the critical regime and beyond}, arxiv.org (2020).


	\bibitem{E18}
	Erignoux, C.: Hydrodynamic limit of boundary driven exclusion processes with nonreversible boundary dynamics, Journal of Statistical Physics, Volume 172, Issue 5,  1327--1357 (2018).
	
	
	\bibitem{ELX18}
	Erignoux, C., Landim, C.,  Xu, T.: Stationary states of boundary driven exclusion processes with nonreversible boundary dynamics, Journal of Statistical Physics, Volume 171, number 4, 599--631 (2018).
	
	
	
	\bibitem{filo} Filo, J.: A nonlinear diffusion equation with nonlinear boundary conditions: methods of lines. Mathematica Slovaca, Volume 38, Issue 3, 273--296, (1988).
	
	
	\bibitem{patriciaantigo}
	Franco, T., Gon\c calves, P., Neumann, A.:	 Hydrodynamical Behavior of Symmetric Exclusion With Slow Bonds, Annnales de l'Institut Henri Poincar\'e - Probabilit\'es et Statistiques, Volume 49, no. 2 (2013)
	
	
	\bibitem{FGS}
	Franco, T., Gon\c calves, P., Schutz, G.: Scaling limits for the exclusion process with a slow site, Stochastic Processes and their applications, Volume 126, Issue 3, 800--831 (2016).
	
	
	\bibitem{G1}
	Gon\c calves, P.: Hydrodynamics for symmetric exclusion in contact with reservoirs,  Stochastic Dynamics Out of Equilibrium
Springer Lecture Notes, Institut Henri Poincar\'e, Paris, France (2019).

	\bibitem{GPV}
	Guo, M. Z., Papanicolaou, G. C., Varadhan, S. R. S.: Nonlinear diffusion limit for a system with nearest neighbor interactions. Commun. Math. Phys. 118, 31--59, (1988).
	\bibitem{KL}
	Kipnis, C., Landim, C.: Scaling limits of interacting particle systems, Springer-Verlag (1999).
	\bibitem{LT}
	Landim, C., Tsunoda, K.: Hydrostatics and Dynamical Large Deviations for a Reaction-Diffusion Model,
	Annals Henri Poincar\'e, Volume 54, no. 1, 51--74 (2018).
	\bibitem{SS}
	Sethuraman, S.,  Shahar, D.: Hydrodynamic limits for long-range asymmetric interacting particle systems,  Elec. J. Probab., Volume  23, no. 130, 1--54 (2018).
	\bibitem{Tsu}
	Tsunoda, K.: Hydrostatic Limit for Exclusion Process With Slow Boundary Revisited,  accepted for publication in RIMS K\^oky\^uroku Bessatsu.
	
	
\end{thebibliography}
\end{document}